\documentclass[]{scrartcl}
\usepackage{amssymb}
\usepackage{amsmath}
\usepackage{amsthm}

\usepackage{enumitem}
\usepackage{xspace}
\usepackage{authblk}
\usepackage{xcolor}
\usepackage{mathabx}
\usepackage{graphicx}

\usepackage{caption}
\usepackage[font=normalsize]{subcaption}

\usepackage[bibstyle=numeric,citestyle=numeric, uniquename=false,sortcites=true,maxbibnames=20, maxcitenames=20,doi=false, url=false, eprint=false, firstinits=true, backend=bibtex, natbib=true, isbn=false, date=year]{biblatex} 
 \bibliography{OCAMbibtex,litATpaper,literature_latin1}

\newcommand{\mbb}[1]{\mathbb #1}

\newcommand{\mcl}[1]{\mathcal #1}

\newcommand{\EMPC}{\textsc{empc}\xspace}
\newcommand{\NMPC}{\textsc{nmpc}\xspace}
\newcommand{\MPC}{\textsc{mpc}\xspace}
\newcommand{\OCP}{\textsc{ocp}\xspace}

\newcommand{\AC}{\textsc{ac}\xspace}
\newcommand{\DC}{\textsc{dc}\xspace}
\newcommand{\OPF}{\textsc{opf}\xspace}
\newcommand{\OCPs}{\textsc{ocp}s\xspace}
\newcommand{\NLP}{\textsc{nlp}\xspace}

\newcommand{\IEEE}{\textsc{ieee}\xspace}
\newcommand{\IPOPT}{\textsc{ipopt}\xspace}
\newcommand{\MATPOWER}{\textsc{matpower}\xspace}
\newcommand{\MATLAB}{\textsc{matlab}\xspace}
\renewcommand{\phi}{\varphi}
\renewcommand{\epsilon}{\varepsilon}

\newcommand{\N}{\mathbb{N}}
\newcommand{\R}{\mathbb{R}}
\newcommand{\U}{\mathbb{U}}
\newcommand{\X}{\mathbb{X}}
\newcommand{\Y}{\mathbb{Y}}
\newcommand{\D}{\mathbb{D}}
\newcommand{\Z}{\mathbb{Z}}

\newcommand{\Rnx}{\mathbb{R}^{n_x}}
\newcommand{\Rny}{\mathbb{R}^{n_y}}
\newcommand{\Rnu}{\mathbb{R}^{n_u}}

\newcommand{\imag}{\mathrm{j}}
\newcommand{\nodVar}{y}
\newcommand{\nodState}{x}
\newcommand{\nodAlgstate}{z}
\newcommand{\nodInput}{u}
\newcommand{\nodDist}{d}

\DeclareMathOperator{\diag}{diag}
\DeclareMathOperator{\interior}{int}
\DeclareMathOperator{\rank}{rank}
\newcommand{\inte}[1]{\interior{\mbb{#1}}}
\newcommand{\eBox}{$\footnotesize\hfill\blacksquare$}

\newtheorem{rema}{Remark}
\newtheorem{lemma}{Lemma}[section]
\newtheorem{corollary}{Corollary}[section]
\newtheorem{proposition}{Proposition}[section]
\newtheorem{theorem}{Theorem}[section]
\newtheorem{assumption}{Assumption}[section]

\newtheorem{definition}{Definition}[section]

\def\keywords{\vspace{0.1em}
	{\textit{Keywords}:\,\relax%
}}

\begin{document}

\author{Timm Faulwasser}
\author{Alexander Engelmann}

\affil{Institute for Automation and Applied Informatics (IAI) \\ 
	   Karlsruhe Institute of Technology (KIT)\\ Hermann-von-Helmholtz-Platz 1 \\
	   76344 Eggenstein-Leopoldshafen, Germany \\
	E-Mail: timm.faulwasser@ieee.org, alexander.engelmann@kit.edu
}

 \title{Towards economic NMPC for multi-stage AC optimal power flow}
\date{}

\maketitle

\begin{abstract}
Recently there has been considerable progress on the analysis of stability and performance properties of so-called economic Nonlinear Model Predictive Control (\NMPC) schemes; i.e. \NMPC schemes employing stage costs that are not directly related to distance measures of pre-computed setpoints.  At the same time, with respect to the energy transition, the use of \NMPC schemes is proposed and investigated in a plethora of papers in different contexts.  
For example receding-horizon approaches  to generator dispatch problems, which is also known as multi-stage Optimal Power Flow (\OPF), naturally lead to economic \NMPC schemes based on non-convex discrete-time Optimal Control Problems (\OCP). 
The present paper investigates the transfer of analytic results available for general economic \NMPC schemes to receding-horizon multi-stage \OPF. 
We propose a blueprint formulation of multi-stage \OPF including \AC power flow equations. Based on this formulation we present results on the dissipativity and recursive feasibility properties of the underlying \OCP. Finally, we draw upon simulations using a $5$ bus system and a $118$ bus system to illustrate our findings.
\end{abstract}
 \keywords{Model predictive control, dissipativity, power systems, dynamic optimal power flow,  economic generator dispatch}

\section{Introduction}
Recently, there has been considerable research progress in analyzing so-called economic Nonlinear Model Predictive Control (\NMPC) schemes based on (system-theoretic) dissipativity assumptions. The appealing promise of economic \NMPC is that the considered stage cost does not need to be related to the distance to specific setpoint as in stabilizing \NMPC \cite{Mayne00a, Rawlings17, Gruene17a}. Rather economic \NMPC allows considering quite generic stage costs. 
While the early works on economic \NMPC \cite{Diehl11a, Angeli12a} rely on specific dissipation inequalities and terminal constraints, it has also been analyzed under which conditions economic \NMPC without terminal constraints yields practical asymptotic stability of the discrete-time closed loop \cite{Stieler14b}, respectively, practical convergence for the continuous-time counterpart \cite{epfl:faulwasser15g}. The crucial observation is that the underlying  dissipation inequality---which relates the stage cost of the Optimal Control Problem (\OCP) with the underlying dynamic system---induces a so-called turnpike property in the open-loop \OCP solution, see e.g. \cite{epfl:faulwasser15h, Gruene13a}. Indeed one can also exploit the turnpike property  to enforce asymptotic stability without (primal) terminal constraints by considering an end penalty which induces a dual/adjoint terminal constraint \cite{kit:zanon18a, kit:faulwasser18e}. Actually one may claim that in the nominal time-invariant setting there exists a mature understanding of dissipativity-based approaches to economic \NMPC; we refer to \cite{kit:faulwasser18c} for a recent and comprehensive literature review.

In power systems research discrete-time \OCPs occur in generator dispatch problems under the label \emph{multi-stage Optimal Power Flow (\OPF)}.\footnote{Note that there is no unified notion for these problems.  In the literature these problems are also referred to as \emph{dynamic optimal power flow}\cite{Xie2001}, \emph{time-constrained optimal power flow}\cite{Meyer-Huebner2015} or \emph{dynamic economic dispatch}\cite{Ross1980}, see \cite{Grainger1994, Frank2012, Frank2016} for tutorial introductions on \OPF problems.} In multi-stage \OPF the aim is to minimize the (monetary) cost of active power generation while satisfying the physical laws of the underlying grid (power flow equations) and operational constraints like generator limits, voltage bounds and line limits, see \cite{Grainger1994, Frank2012, Frank2016,kit:faulwasser18d}. The main challenge of \OPF problems is twofold: practically relevant grid models can easily comprise several hundreds to thousands of nodes and the power flow equations constitute a set of nonlinear equality constraints.
In the past, most \OPF problems (sometimes also called economic dispatch problems in case of simplified grid models) have been considered decoupled in time (i.e. single stage) as the influence of available storage (pumped-hydro plants etc.) has been negligible. Nowadays, 
the volatility of renewable energy production induces challenges for efficient energy system operation; e.g. violations of generator ramp limits imply the need to investigate novel control strategies. Moreover, the amount of energy storage present in the grid in various forms (batteries, electric vehicles, thermal storage) is increasing rapidly. Hence, there is a need to investigate time-coupled multi-stage \OPF problems and \NMPC solutions thereof.  

There exists a plethora of papers and results on discrete-time optimal control and on model predictive control of power and energy systems. In general one can distinguish two main lines of research on \MPC for power systems:  
\begin{enumerate}
\item \MPC for $\{$energy management, generator dispatch, etc.$\}$ using price-based objectives and considering rather slow time-scales ($15$ minutes -- $1$ hour); \item and \MPC targeting voltage and frequency stabilization on rather fast time-scales (a few seconds). 
\end{enumerate}

As for large grids \OPF problems are already challenging in the single-stage case, early works on multi-stage \OPF usually consider  linearized grid physics (i.e. \DC power flow equations) at the cost of losing reactive power and voltage information.  For example \cite{Ross1980,Wood1982} use generator ramp constraints and solve multi-stage \OPF via  dynamic programming, while Lagrangian relaxation techniques for time-wise decomposition are employed in \cite{Batut1992}.
Triggered by the need of handling voltage and reactive power injection limits, follow-up works consider the full \AC grid model including the highly nonlinear power flow equations \cite{Chen2005,Olivares2014,Olivares2015}.  A similar approach is used in \cite{Kennel2013}, where a hierarchical \MPC scheme is used for EV charging and concurrent frequency control; in \cite{Ulbig2011} a three-layer scheme considering transmission planning, power dispatch and frequency regulation is proposed.
A recent review on dispatch via multi-stage \OPF can be found in \cite{Xia2010}. 
Economic \MPC for multi-energy systems (considering thermal energy) including storage can be found in \cite{Arnold2010,Stoyanova2018}.
Furthermore, multi-stage \OPF can be used for cost-optimal energy management of microgrids based on price signals using a linearized \DC model \cite{Braun2017}. While \cite{kit:appino17b} discusses the effect of uncertain forecasts on storage scheduling, a distributed approach to scheduling is suggested in \cite{kit:braun18a}. The recent paper \cite{Meyer-Huebner2017} considers \MPC for storage scheduling in high-voltage grids. Stochastic economic \NMPC of micro-girds is proposed in \cite{Vasilj17a}. 
An approach to combined frequency regulation and economic dispatch via economic \MPC, including stability analysis, with \DC grid models is presented in  \cite{Koehler2017}. 

With respect to branch (ii) several works \cite{Hiskens2005,Hovgaard2010,Zima2005,Riverso2015}  consider voltage or frequency stabilization; frequency control and automatic generation control\cite{Venkat2008} as well as the use of quadratic stability constraints ensuring stability (which lead to state path constraints instead of terminal constraints)\cite{Tran2014} have been investigated.  Approaches to consider voltage control via minimizing the absolute value of total reactive power injection are presented in \cite{Zima2005}. A hierarchical economic \MPC scheme considering generator cost and electricity prices for power from an upper grid level is proposed in \cite{Hans2018}. In \cite{Bansal2014} a multi-stage voltage control scheme for electric vehicle charging  ensuring recursive feasibility and exponential stability is proposed. Scheduling based on linearized power-flow equations and set-point tracking is investigated in \cite{Almassalkhi2015,Almassalkhi2015a}. Common to all these approaches is that stabilizing/tracking \NMPC formulations are employed and that the power grids is approximated in a linearized fashion (i.e. as \DC-\OPF).

The present paper investigates receding-horizon solutions to multi-stage \OPF problems with economic cost functions; i.e. we consider an economic \NMPC approach to multi-stage \AC \OPF. While there has been considerable interest in structured \MPC design for voltage and frequency control,  the works on  \NMPC multi-stage \OPF are mostly agnostic to the recent progress on economic \NMPC. Moreover, the large majority of works relies on convex \DC approximations of the underlying power grid. The present paper aims at partially closing this gap; i.e. we try to transfer and verify the dissipativity notions established in economic \NMPC to so-called multi-stage \OPF problems while explicitly considering  the nonlinear \AC power-flow equations as equality constraints. To this end and extending \cite{kit:faulwasser18d}, we will provide a  blueprint \OCP formulation of multi-stage \AC \OPF that enables a formal analysis. Here, we go beyond \cite{kit:faulwasser18d} by including energy storage and by showing how to handle the power-flow equations by means of projection. We then analyze the dissipativity properties of this \OCP for constant  power injections. Moreover, we discuss recursive feasibility in case of varying disturbances. 
To the best of the authors' knowledge, the present paper is the first one explicitly establishing dissipativity properties for \OCPs arising from multi-stage \AC \OPF. 

The remainder of this paper is structured as follows: Section \ref{sec:EMPC} recalls the essentials of dissipativity-based approaches to economic \NMPC; Section \ref{sec:problem} introduces the multi-stage \OPF problem as a discrete-time \OCP. Section \ref{sec:results} analyses the recursive feasibility and dissipativity properties of this \OCP; moreover the properties of the closed \NMPC loop are discussed. 
Section \ref{sec:example} draws upon two examples to illustrate the findings. The paper ends with outlook and conclusions in Section \ref{sec:conclusion}.

\subsection*{Notation}
We denote the state and control of a system, respectively, as $x\in \mathbb{R}^{n_x}$ and $u\in \mathbb{R}^{n_u}$.
We define $\mathbb{I}_{[a,b]} := \{a,\ldots,b\}$, $a,b \in \mathbb{Z}$, i.e. integers.
The matrix $I^{n_x}$ 
 is the identity matrix of $ \R^{n_x\times n_x}$, while $0^{n_x}$ 
 is the zero matrix of  $\R^{n_x\times n_x}$.
Subsets of $\R^{n_x}$ are denoted by $\X$; the interior of $\X$ is denoted by $\interior(\X)$. The pointwise image of a set $\X\subseteq \Rnx$ under a map $h:\Rnx \to \R^{n_y}$ is written as $h(\X)=\Y$, and the pre-image of $\Y$ under $h$---i.e. the set $\{x\ | \ h(x) \in \Y\}$---is denoted by $h^{-1}(\Y)$.
The concatenation of $x \in \Rnx$ and $u\in\Rnu$ is written as $\begin{bmatrix}x &u\end{bmatrix}^\top$.

\section{Preliminaries -- Dissipativity and economic NMPC} \label{sec:EMPC}
We consider time-invariant discrete-time systems described by  \vspace*{1.5mm}
\begin{equation} \label{eq:sys_def}
x(t+1) = f(x(t), u(t)),\qquad x(0) = x_0 \in \X_0,  \vspace*{1.5mm}
\end{equation}
where $x \in \Rnx$ is the state, $u \in \Rnu$ is the input, $f: \Rnx \times \Rnu \to \Rnx$ denotes the continuous state transition map, and
$t\in \mbb{Z}$ is the discrete time variable.
 States and inputs are assumed to be restricted by the compact sets $\X \subset \Rnx$ and $\U \subset \Rnu$, respectively. 
Corresponding to system \eqref{eq:sys_def}, one considers a cost functional  \vspace*{1.5mm}
\begin{equation} \label{eq:cost_functional}
J_N(x_0, u(\cdot)) = \sum_{k = 0}^{N-1} \ell(x(k), u(k)) \vspace*{1.5mm}
\end{equation}
 which models the performance requirements of \eqref{eq:sys_def} with  the continuous stage cost  $\ell: \X \times\U \to \mbb{R}$.
In general, an \NMPC scheme without terminal constraints is based on solving 
the following finite-horizon discrete-time \OCP at each time step $t=0,1,2,\ldots$:
\begin{subequations} \label{eq:EMPC_OCP}
    \begin{align}
    V_N(x(t)) = \min_{u(\cdot|t)} &\sum_{k = 0}^{N-1} \ell(x(k|t), u(k|t))  \label{eq:EMPC_OCP:obj} \\
    \text{subject to}&  \nonumber\\
     x(k+1|t) &= f(x(k|t), u(k|t)), \quad x(0|t) = x(t),\quad &k\in \mathbb{I}_{[0,N-1]}  \label{eq:EMPC_OCP:dyn}\phantom{.}\\
    \begin{bmatrix} x(k|t) &u(k|t)\end{bmatrix}^\top&\in \X\times\U,  \hspace*{4.05cm}	&k\in \mathbb{I}_{[0,N-1]} . \label{eq:EMPC_OCP:con}  \vspace*{3.5mm}
  \end{align}
  \end{subequations}
  Here, $N \in \mbb{N}$ is the prediction horizon and $V_N(x(t))$ is the optimal value function of \eqref{eq:EMPC_OCP}. 
Equations \eqref{eq:EMPC_OCP:dyn}--\eqref{eq:EMPC_OCP:con} summarize the equality constraints imposed by the dynamics and additional constraints on states and inputs, which are typically described by inequalities. As \eqref{eq:EMPC_OCP} is essentially a Nonlinear Program (\NLP), we require continuity of $f$ and $\ell$. In case the feasible set is non-empty our assumptions imply that an optimal solution to problem~\eqref{eq:EMPC_OCP} exists, see \citep{Bertsekas99a}.  

The superscript $(\cdot)^\star$ indicates variables related to optimal solutions  of \eqref{eq:EMPC_OCP}.
Furthermore, in order to distinguish predicted variables from closed-loop variables, we use the notation $\cdot (k|t)$ to denote $k$-step ahead predictions computed at time $t\in \Z$ based on the current (real) system state $x(t)$.
For example, we write $u^\star(k|t)$ to refer to the $k$th element of the optimal predicted input sequence to \OCP \eqref{eq:EMPC_OCP} computed for the initial condition $x(t)$, and we denote the corresponding optimal state trajectory by $x^\star(\cdot|t)$. 
Hence, one defines the \NMPC feedback as  \vspace*{1.5mm}
\[\mu_N(x(t)) := u^\star(0 | t), \vspace*{1.5mm}\]
 i.e., as the first element of the optimal input sequence, and obtains the next state of the closed loop system as  \vspace*{1.5mm}
\begin{equation} \label{eq:sys_closed_loop}
x(t+1) = f(x(t), \mu_N(x(t))),\qquad x(0) = x_0.  \vspace*{1.5mm}
\end{equation}
Throughout this paper we will not consider any plant-model mismatch, i.e., we assume that $f$ in \eqref{eq:EMPC_OCP:dyn} and in \eqref{eq:sys_closed_loop} are identical. 

Frequently the analysis of the closed-loop system \eqref{eq:sys_closed_loop} is based on the following dissipativity notion originally suggested by \citep{Angeli12a}.
\begin{definition}[Strict dissipativity]\label{def:DI}~\\
\begin{enumerate}[label=(\roman*)]\vspace*{-.5cm}
\item System \eqref{eq:sys_def} is said to be \emph{dissipative with respect to the steady-state pair $(x_s, u_s) \in\X\times\U$}, 
if there exists a non-negative function $\lambda:\X \to \R$ such that for all $x\in \X, u\in \U$
\begin{subequations} \label{eq:DI}
\begin{equation}\label{eq:DI_non_str}
\lambda(f(x, u)) -\lambda(x) \leq \ell(x, u) - \ell(x_s, u_s).
\end{equation}
\item If, additionally, there exists $\alpha_\ell\in\mcl{K}_\infty$ such that
\begin{equation}  \label{eq:DI_str}
\lambda(f(x, u)) -\lambda(x) \leq -\alpha_\ell\left(\left\|(x-x_s)\right\|\right) + \ell(x, u) - \ell(x_s, u_s).
\end{equation}
\end{subequations}
then  \eqref{eq:sys_def} is said to be \emph{strictly dissipative with respect to $(x_s, u_s)$}. 

\item If, for all $N\in\N$ and all $x_0 \in \X_0$, the dissipation inequalities \eqref{eq:DI} hold along any optimal pair of \OCP  \eqref{eq:EMPC_OCP}, 
then \emph{\OCP \eqref{eq:EMPC_OCP} is said to be (strictly) dissipative with respect to $(x_s, u_s)$}.   \eBox
\end{enumerate}
\end{definition}

We remark that $\ell$ in  \eqref{eq:DI} is the stage cost of \OCP \eqref{eq:EMPC_OCP}. Denoting $s:\X\times\U \to \R$
\[
s(x,u) :=  \ell(x, u) - \ell(x_s, u_s) 
\]
as a supply rate and $\lambda$ in \eqref{eq:DI} as a storage function, it is clear that \eqref{eq:DI} are dissipation inequalities, see \citep{Moylan14a, Willems72a, Willems07a}. We remark that the original dissipativity concept proposed by Jan Willems \citep{Willems72a} allows for an intuitive energy-related interpretation. In the context of the present paper, however, dissipativity is an abstract system-theoretic concept used to analyze \OCPs. As such it is not directly related to the actual dissipation of energy. Moreover, it is straightforward to show that dissipativity implies that the steady state pair $(x_s, u_s)$ is a globally optimal solution to the steady-state optimization problem
\[
\min_{x, u} ~\ell(x, u) \quad \text{ subject to } x = f(x, u) \text{ and } \begin{bmatrix} x & u\end{bmatrix}^\top \in \X\times\U.
\]
Finally, we remark that in the literature on economic \NMPC different variants of the inequality \eqref{eq:DI_str} are considered; i.e. occasionally strictness in $x$ and $u$ is required \cite[Rem. 3.1]{kit:faulwasser18c}. For the purposes of the present paper, however, it suffices to consider strictness with respect to $x$ only. 

Next, we summarize conditions under which the \NMPC scheme defined by  \eqref{eq:EMPC_OCP} yields practical asymptotic stability  of the closed-loop system \eqref{eq:sys_closed_loop}.
\begin{assumption}[Reachability and dissipativity]
\label{ass:EMPC}~\\
\begin{enumerate}[label=(\roman*)]\vspace*{-.5cm}
\item \OCP \eqref{eq:EMPC_OCP} is strictly dissipative with respect to $(x_s, u_s) \in \X\times\U$ in the sense of Definition \ref{def:DI} (iii). 
\item For all $x_0 \in \X_0$, there exists an infinite-horizon admissible input $u(\cdot;x_0)$, $c \in (0,\,\infty)$, $\rho\in [0,1)$, such that
\[
\|(x(k;x_0, u(\cdot;x_0)), u(k;x_0)) - (x_s, u_s)\| \leq c\rho^k,
\]
i.e. the steady state $x_s$ is exponentially reachable.
\item The Jacobian linearization of system \eqref{eq:sys_def} at  $(x_s, u_s) \in \interior \left(\X\times\U\right)$ is  $n_x$-step reachable.\footnote{Recall that $n_x$-step reachability of $x^+ = Ax +Bu$ implies that starting from $x=0$ one can reach any $x\in\Rnx$ within $n_x$ time steps; and one can steer any  $x \not=0$ to the origin within $n_x$ time steps, cf. \citep{Weiss72a}. In other words, $n_x$-step reachability implies $n_x$-step controllability.} \eBox
\end{enumerate}
\end{assumption}
\begin{theorem}[Practical stability of \EMPC without terminal constraints] \label{thm:stab_EMPC}
Let Assumption \ref{ass:EMPC} (i--iii) 
hold and suppose that $\X$ is compact. 
Then, there exists a sufficiently large horizon $N\in\N$, such that the closed-loop system \eqref{eq:sys_closed_loop} arising from the receding horizon solution to \OCP \eqref{eq:EMPC_OCP} has the following properties:
\begin{enumerate}[label=(\roman*)]
\item If, for the horizon $N\in\N$,  \OCP \eqref{eq:EMPC_OCP} is feasible for $t = 0$ and $x(0) \in \X_0$, then it is feasible for all $t \in \N$.
\item There exist $\rho \in \R^+$ and $\beta \in \mcl{K}\mcl{L}$ such that, for all $x(0) \in \X_0$, the closed-loop trajectories generated by  \eqref{eq:sys_closed_loop} satisfy
\[
\|x(t) - x_s\| \leq \max\{\beta(\|x(0) - x_s\|,\,t),\, \rho\}. 
\vspace*{-.65cm}
\]   
\eBox
\end{enumerate}
\end{theorem}
The proof of this result is based on the fact that the dissipativity property of  \OCP \eqref{eq:EMPC_OCP} implies turnpike properties of \OCP \eqref{eq:EMPC_OCP}; it can be found in \cite{kit:faulwasser18c}. Earlier versions presented in \cite{Gruene17a, Gruene13a} do not include the recursive feasibility statement. We remark that under additional continuity assumptions on the storage function $\lambda$ and on the rotated optimal value function one can show that the size of the neighborhood---i.e. $\rho$ in (ii)---converges to $0$ as $N$ goes to $\infty$, cf. \cite[Lem. 4.1 and Thm. 4.1]{kit:faulwasser18c}. Similar results can also be established for the continuous-time case \cite{epfl:faulwasser15g, kit:zanon18a}.

Observe that in Theorem \ref{thm:stab_EMPC}  it is required that $(x_s, u_s) \in \interior \left(\X\times\U\right)$, i.e. the  optimal steady state pair lies in the interior of the constraints. The reason for this requirement is that the recursive feasibility construction exploits local controllability,  cf. proof of Prop. 4.2 in \cite{kit:faulwasser18c}. 
If, however, one can ensure that with inputs $u(\cdot)\in \U$  system\eqref{eq:sys_def} is finite-time reachable on $\mcl{B}_\rho(x_s) \cap \X$ for some $\rho >0$---whereby $\mcl{B}_\rho(x_s)$  is an open ball of radius $\rho$ centered at $x_s$---, then one can relax  $(x_s, u_s) \in \interior \left(\X\times\U\right)$ to  $(x_s, u_s) \in \X\times\U$. This is summarized next.
 
\begin{corollary} \label{cor:stab_EMPC}
Let Assumption \ref{ass:EMPC} (i--ii) 
hold and suppose that $\X$ is compact. Moreover, suppose that  \eqref{eq:sys_def} is finite-time reachable on $\mcl{B}_\rho(x_s(d)) \cap \X$ for some $\rho >0$ and with inputs $u(\cdot)\in \U$.
Then the statements of Theorem \ref{thm:stab_EMPC} hold. \eBox
\end{corollary}

The application of Theorem \ref{thm:stab_EMPC} requires to verify the dissipativity properties of the underlying \OCP, cf. Definition \ref{def:DI}. Except for special cases---i.e. linear dynamics, quadratic stage costs cf. \cite{Gruene18a,Zanon14a}---the computation of storage functions $\lambda$ is as hard as the computation of Lyapunov functions \cite{epfl:faulwasser15h,Ebenbauer06a}. 
The next technical result reports an observation which is of interest in its own right and which will be helpful in verifying dissipativity properties in the context of this paper. 
\begin{definition}[Controlled forward invariant set]
A set $\tilde{\X} \subset \Rnx$ is said to be controlled forward invariant for system \eqref{eq:sys_def} and input constraint $\U$ if for all
$x \in \tilde{\X}$ there exists $u \in \U$ such that $f(x, u) \in \tilde{\X}$.   \eBox
\end{definition}

\begin{lemma}[Dissipativity on subsets of state constraints]\label{lem:DI_subset}
Let system \eqref{eq:sys_def} be (strictly) dissipative in the sense of Definition \ref{def:DI} (i--ii) on some compact constraint sets $\X$ and $\U$, where $\X$ is controlled forward invariant.
Then, for any subset $\tilde{\X}\subseteq \X$, which is controlled forward invariant with respect to $\tilde{\U}\subseteq\U$, system \eqref{eq:sys_def} is also (strictly) dissipative in the sense of Definition \ref{def:DI}  (i--ii) on  $\tilde{\X}$ and $\tilde\U$.   \eBox
\end{lemma}
\begin{proof}
Recall that, for any supply rate $s:\X\times\U\to\R$, dissipativity can equivalently be characterized in terms of the available storage; i.e. in terms of the boundedness of the optimal value function $\lambda_a:\X\to\R$ of the following free-end-time \OCP
\begin{subequations} \label{eq:OCP_DI}
\begin{align}
    \lambda_a(x) = \sup_{u(\cdot), N} &\sum_{k = 0}^{N-1} -s(x(k),u(k))   \\
    \text{subject to}&  \nonumber\\
     x(k+1) &= f(x(k), u(k)), \quad x(0) = x, \quad &k\in \mathbb{I}_{[0,N-1]}\phantom{.}   \\
   \begin{bmatrix}x(k) &u(k)\end{bmatrix}^\top&\in \X\times\U,  \hspace*{3.05cm}	& k\in \mathbb{I}_{[0,N-1]}.  
\end{align}
\end{subequations}
More precisely, system \eqref{eq:sys_def} is dissipative on $\X\times\U$ if and only if $\lambda_a(x) \leq \bar \lambda < \infty$ for all $x \in \X$, see \cite{Willems72a} or \cite{Mueller14a} for details. 

Suppose that system \eqref{eq:sys_def} is not (strictly) dissipative on $\tilde{\X}$ and $\tilde{\U}$. Then, since $\tilde{\mathbb{X}}$ is controlled forward invariant, there exists $\tilde x \in \tilde{\X}$ and a sequence $\tilde{u}(\cdot)\in \tilde\U$ such that $\lambda_a(\tilde x) = \infty$. Since $\tilde{\X}\subseteq \X$ it is clear that  $\tilde x \in {\X}$. Hence $\lambda_a(\tilde x) = \infty$ contradicts (strict) dissipativity of  \eqref{eq:sys_def} on $\X$ and $\U$. 
\end{proof}
We remark that one can easily extend the previous result to the case of a dissipative \OCP whereby the initial conditions are restricted to a controlled forward invariant set, while the complete state constraint $\X$ does not need to be controlled invariant.

\section{Multistage AC Optimal Power Flow} \label{sec:problem}
The present paper aims at transferring the dissipativity-based framework for economic \NMPC summarized in Theorem \ref{thm:stab_EMPC} to generator dispatch problems (or multi-stage  \OPF problems) arising in power systems. To this end and similar to \cite[Sec. 2]{kit:faulwasser18d}, we consider balanced electrical \AC grids at steady state modeled by  $(\mathcal{N},\mathcal{G},\mathcal{S},Y)$, where $\mathcal{N}=\{1, \hdots, N\}$ is the set of buses (nodes), $\mathcal{G} \subseteq \mathcal{N}$ is the non-empty set of generators, $\mathcal{S} \subseteq \mathcal{N}$ is the  set of storages/batteries, and $Y = G + \imag B \in \mathbb{C}^{N\times N}$ is the bus admittance matrix \cite{Grainger1994}. The off-diagonal entries of ${Y}$ can be written as $\textcolor{black}{-}y_{lm}= {g}_{lm} + \imag{b}_{lm}$, whereby ${g}_{lm}$ is the conductance for the line $lm$, respectively, ${b}_{lm}$ is the line susceptance. The diagonal entries of $Y$ are $y_{ll} = y_l + \sum_{l\not = m} y_{lm}$, where $y_l$ accounts for the ground (shunt) admittance connected to bus $l$.

Implicitly, we assume symmetric three-phase \AC conditions.  Thus, every bus $l \in \mathcal{N}$ is described by its voltage phasor $v_l \mathrm{e}^{k \theta_l} \in \mathbb{C}$ and net apparent power $s_l = p_l + j q_l \in \mathbb{C}$, or equivalently by its voltage magnitude $v_l$, voltage phase $\theta_l$, net active power $p_l$, and net reactive power $q_l$.

In the considered setting the steady-state behavior of the power grid is described by the so-called power flow equations 
\begin{subequations}
	\label{eq:powerflow}
	\begin{align}
	p_l = &\;v_l\sum_{m \in \mathcal{N}} v_m\left(G_{lm}\cos(\theta_{lm})+B_{lm}\sin(\theta_{lm})\right),\\
	q_l = &\;v_l\sum_{m \in \mathcal{N}} v_m\left(G_{lm}\sin(\theta_{lm})-B_{lm}\cos(\theta_{lm})\right),
	\end{align}
\end{subequations}
where the shorthand notation $\theta_{lm} := \theta_l -\theta_m$ is applied.
Observe that in the  power flow equations~\eqref{eq:powerflow} the phase angles $\theta_l$ occur as pair-wise differences, therefore one bus $l_0\in \mathcal{N}$ is specified as so-called reference (slack) bus 
$	\theta_{l_0} = 0$ for $	l_0 \in \mcl{N}$; w.l.o.g. we will assume $l_0 = 1$ in the remainder.

For the sake of simplicity, we consider only one generator per bus (i.e. $\mcl{G} \subseteq \mcl{N}$), respectively, one storage per bus (i.e. $\mcl{S} \subseteq \mcl{N}$).  
We describe the net apparent power of bus $l \in \mathcal{N}$ by
\begin{equation*}
s_l = p_l + \imag q_l =
 p_l^{\text{g}} +  p_l^{\text{s}} - p_{l}^{\text{d}} + \imag \left( q_{l}^{\text{g}} + q_{l}^{\text{s}} - q_{l}^{\text{d}}\right) 
\end{equation*}
where $p_{l}^{\text{g}}$, $q_{l}^{\text{g}}$ are controllable power injections for all generator nodes $l \in \mathcal{G}$, $p_{l}^{\text{s}}$, $q_{l}^{\text{s}}$ are controllable power injections for all storage nodes $l \in \mathcal{S}$, and $p_{l}^{\text{d}}$, $q_{l}^{\text{d}}$ are uncontrollable power sinks/sources for all $l \in \mathcal{N}$. If $l\not\in\mcl{S}$ then $p_{l}^{\text{s}}=0$, $q_{l}^{\text{s}}=0$, and if $l\not\in\mcl{G}$ then $p_{l}^{\text{g}}=0$, $q_{l}^{\text{g}}=0$.  

For the sake of compact notation, we define the auxiliary variable $\nodVar \in \mathbb{R}^{n_{\nodVar}}$, the disturbance $\nodDist \in \mathbb{R}^{n_{\nodDist}}$, and the (algebraic) state $\nodAlgstate \in \mathbb{R}^{n_{\nodAlgstate}}$ as follows
\begin{subequations}
\begin{align} 
\nodVar &=	\begin{bmatrix} p_{l}^{\text{g}} & q_{l}^{\text{g}} &  p_{l}^{\text{s}} & q_{l}^{\text{s}}\end{bmatrix}_{l\in\mcl{G}\cup\mcl{S}}^\top \in \mbb{R}^{n_{\nodVar}},  && n_{\nodVar} = 2 | \mathcal{G} | + 2|\mcl{S}|, \label{eq:NodVar} \\
\nodDist &=	\begin{bmatrix} p_{l}^{\text{d}}& q_{l}^{\text{d}} \end{bmatrix}^\top_{l\in\mcl{N}}  \in \mbb{R}^{n_{\nodDist}}, && n_{\nodDist} = 2 | \mcl{N} |, \label{eq:NodDist}\\
\nodAlgstate &= \begin{bmatrix}v_l & \theta_l \end{bmatrix}_{l\in\mcl{N}}^\top \in \mbb{R}^{n_{\nodAlgstate}},  && n_{\nodAlgstate} = 2 | \mcl{N} |. \label{eq:NodAlgstate} 
\end{align}
The variable $\nodVar$ collects all active and reactive power injections that can be controlled/ manipulated; $\nodDist$ is the vector of uncontrolled loads (injections of renewables or demands); and $\nodAlgstate$ collects the phase angles and voltage magnitudes at all buses. 
The active power injections of the generators are typically subject to ramp constraints, i.e.  $ |p_{l}^{\text{g}}(k) -  p_{l}^{\text{g}}(k-1)| \leq \overline{\delta p}_l^{\text{g}}$. In other words,
the generators are subject to the simple linear dynamics (with input constraints)
\begin{align*}
p_{l}^{\text{g}}(k+1) &= p_{l}^{\text{g}}(k) + \delta p_{l}^{\text{g}}(k),  \quad l \in \mcl{G}, \quad  |\delta p_{l}^{\text{g}}(k)| \leq \overline{\delta p}_l^{\text{g}} \\
q_{l}^{\text{g}}(k+1) &= q_{l}^{\text{g}}(k) + \delta q_{l}^{\text{g}}(k), \quad l \in \mcl{G}.
\end{align*}
Note that while the incremental active power injections, $\delta p_{l}^{\text{g}}(k)$, are typically constrained, the incremental reactive power injections, $\delta q_{l}^{\text{\textcolor{blue}{s}}}(k)$, are not subject to constraints.

 Moreover, the storages are of limited capacity; i.e. they are subject to the discretized dynamics
  \[
   e_{l}(k+1) =  e_{l}(k) + \Delta\cdot p_{l}^{\text{s}}(k),\quad l\in\mcl{S},
  \]
  where $e_l(k) \in [0, \overline{e}_l]$ is the state of charge of the storage $l$ at time $k$ and $\Delta>0$ refers to the sampling period. Observe that, for the sake of simplicity, we do not consider charging losses. Moreover, note that the reactive power injections of a battery storage, $q_{l}^{\text{s}}$,  have no effect on the state of charge.  To compactly summarize these dynamics we introduce the following state variable
\begin{align}
\nodState =  \begin{bmatrix}\begin{bmatrix} p_{l}^{\text{g}} & q_{l}^{\text{g}} \end{bmatrix}_{l\in\mcl{G}} &\begin{bmatrix}  e_{l}\end{bmatrix}_{l\in\mcl{S}}\end{bmatrix}^\top \in \mbb{R}^{n_{\nodState}},  && n_{\nodState} = 2 | \mathcal{G} | + |\mcl{S}|. \label{eq:NodState} 
\end{align}
The inputs $\nodInput$ driving the state variables are: the increments of active and reactive generator powers for generator buses ($\delta p_l^{\text{g}}, \delta q_l^{\text{g}}, l\in\mcl{G}$) and the active charging/discharging power for storage buses and the reactive power  provided by storage buses ($p_{l}^{\text{s}}, q_{l}^{\text{s}},l \in \mcl{S}$); i.e.
\begin{align}
\nodInput = \begin{bmatrix}\begin{bmatrix} \delta p_{l}^{\text{g}} & \delta q_{l}^{\text{g}} \end{bmatrix}_{l\in\mcl{G}} & \begin{bmatrix}  p_{l}^{\text{s}} & q_{l}^{\text{s}}\end{bmatrix}_{l\in\mcl{S}}\end{bmatrix}^\top \in \mbb{R}^{n_{\nodInput}},  && n_{\nodInput} = 2 | \mathcal{G} | + 2|\mcl{S}|. \label{eq:NodInput} 
\end{align}
\end{subequations}
Summing up, the dynamics of $\nodState$ are given by 
\begin{subequations} \label{eq:sys_OPF}
\begin{align}
\nodState(k+1) &= A\nodState(k) + B\nodInput(k),\\
\nodVar(k) & =C\nodState(k) + D\nodInput(k) \label{eq:sys_OPF_output}
\end{align}
where 
\begin{equation}
\begin{aligned}
A &= I^{n_\nodState}, \, 
&B &=
\begin{bmatrix}
 I^{2|\mcl{G}|\phantom{\times 2|\mcl{G}|}} & 0^{2|\mcl{G}|\times 2|\mcl{S}|}\\
0^{|\mcl{S}|\times 2|\mcl{G}|} &   ( \Delta\cdot I^{|\mcl{S}|}\; 0^{|\mcl{S}|})
\end{bmatrix}
\,, \\
C &= 
\begin{bmatrix}
I^{2|\mcl{G}|\phantom{\times 2|\mcl{G}|}} & 0^{{2|\mcl{G}|\times}|\mcl{S}|}\\
 0^{2|\mcl{S}|{\times 2|\mcl{G}|}} &  0^{2|\mcl{S}|\times |\mcl{S}|}
\end{bmatrix},
\, &D &= \diag\left(0^{2|\mcl{G}|},~  I^{2|\mcl{S}|}\right).
\end{aligned}
\end{equation}
\end{subequations}
Recall that the reactive power injections provided by the storages, $q_{l}^{\text{s}},l \in \mcl{S}$, do not directly influence the state of charge $e_l$. Thus in the rectangular matrix $B$ the block $\left(\Delta\cdot I^{|\mcl{S}|},~ 0^{|\mcl{S}|}\right)$ appears.

It deserves to be noted that the input, state, and auxiliary variables are typically subject to box constraints, i.e. 
\begin{align*}
 \nodInput &\in  \U,\quad \nodState\in\X \quad  \nodAlgstate \in \Z,\\
\U&:=\bigtimes_{l\in\mcl{G}\cup\mcl{S}} ~\left[-\overline{\delta p}_l^{\text{g}},\, \overline{\delta p}_l^{\text{g}}\right] \times \R\times \left[\underline{p}_l^{\text{s}},\, \overline{p}_l^{\text{s}}\right]\times  \left[\underline{q}_l^{\text{s}},\, \overline{q}_l^{\text{s}}\right]
\end{align*}
 Furthermore, our choice of variables allows writing the power flow equations~\eqref{eq:powerflow} in terms of a system of nonlinear algebraic equations
\begin{equation}
\label{eq:PFE_compact}
F: \mbb{R}^{n_{\nodVar}} \times \mbb{R}^{n_{\nodAlgstate}} \times \mbb{R}^{n_{\nodDist}} \rightarrow \mbb{R}^{2 |\mcl{N}|} \qquad  F(\nodVar,\nodAlgstate; \nodDist) = 0,
\end{equation}
where the semicolon notation emphasizes the dependency on the exogenous disturbance $\nodDist$.

Now we are ready to formulate the multi-stage Optimal Power Flow (\OPF) problem as a discrete-time \OCP as follows:
\begin{subequations} \label{eq:OPF_OCP_naive}
    \begin{align}
    \min_{\nodInput(\cdot|t)} &\sum_{k = 0}^{N-1} \ell(\nodState(k|t), \nodInput(k|t))  \\
    \text{subject to}&  \nonumber\\
     \nodState(k+1|t) &= A\nodState(k|t) + B\nodInput(k|t)), \; \nodState(0|t) = \nodState(t),\; &k\in \mathbb{I}_{[0,N-1]} , \\
     \nodVar(k|t) & =C\nodState(k|t) + D\nodInput(k|t),  \hspace*{1cm} &k\in \mathbb{I}_{[0,N-1]} , \\
     0 & = F(\nodVar(k|t),\nodAlgstate(k|t); \nodDist(k|t)),  \hspace*{1cm}&k\in \mathbb{I}_{[0,N-1]}  , \label{eq:OPF_OCP_PF}\\
    \left[\nodState(k|t) \;\nodInput(k|t)\;\nodVar(k|t)\; \nodAlgstate(k|t)\right ]^\top&\in \X\times\U\times\Rny\times\Z,  \hspace*{1cm}	 &k\in \mathbb{I}_{[0,N-1]} . 
  \end{align}
  \end{subequations}
Observe that the stage cost does not depend on  $\nodAlgstate$ since it only penalizes economic costs related to $\nodState$ and $\nodInput$ (i.e. active power generation). Moreover, note that while the dynamics are a linear system with feed-through, the \AC power flow equations appear in \eqref{eq:OPF_OCP_PF}. In other words, 
the dynamic constraints of this \OCP are given by implicit difference equations (i.e. the discrete-time counterpart to differential algebraic equations).  We remark that the implicit dynamics above are \textit{not} a discretization of the power-swing differential-algebraic equations. Indeed for the most part the mulit-stage OPF problem is a quasi-stationary one, hence the  implicit dynamics above can be regarded as a quasi-stationary model. 

While locally one may employ the implicit function theorem to solve the implicit equations  \eqref{eq:OPF_OCP_PF} for $\nodAlgstate$,\footnote{Given $\nodDist_0$, $\nodVar_0$ and $\nodAlgstate_0$ with $F(\nodVar_0,\nodAlgstate_0; \nodDist_0)=0$ and $\nodDist$, $\nodVar$ close to $\nodDist_0$, $\nodVar_0$ one may solve \eqref{eq:PFE_compact} for $\nodAlgstate$. Note that due to the fact that the power flow equations are formulated in terms of phase-angle difference $\theta_{ml}$ a rank-deficient Jacobian occurs, which has to be taken care of (e.g. defining the phase at  the slack bus as $\theta_{l_0}=0$).} it is known in power engineering that 
in general the power-flow equations may have non-unique solutions \cite{Nick18}. Hence a rigorous and direct analysis of \OCP \eqref{eq:OPF_OCP_naive} would require to deal with the tedious technicalities of set-valued implicit dynamics. 

Here, however, we propose a different approach. Observe that the algebraic state $\nodAlgstate$ does not appear in the stage cost nor does it directly influence the dynamics of $\nodState$.  
Hence the constraints $F(\nodVar,\nodAlgstate; \nodDist)=0$ and $\nodAlgstate\in\Z$ can be expressed more conveniently via the non-convex set
\begin{equation}\label{eq:Yset}
\Y(\nodDist) := \left\{\nodVar \in \R^{n_\nodVar}\,|\, F(\nodVar,\nodAlgstate; \nodDist)=0 \quad \nodAlgstate\in\Z\right\}.
\end{equation}
Note that $\Y(\nodDist)$ can be understood as a subset of the projection of the so-called power-flow manifold---i.e. the solution set to the power flow equations \eqref{eq:powerflow}---onto $\R^{n_\nodVar}$.
 This allows reformulating \OCP \eqref{eq:OPF_OCP_naive} from above as:
 \begin{subequations} \label{eq:OPF_OCP}
    \begin{align}
    V_N(\nodState(t), t) = \min_{\nodInput(\cdot|t)} &\sum_{k = 0}^{N-1} \ell(\nodState(k|t), \nodInput(k|t))  \\
    \text{subject to}&  \nonumber\\
     \nodState(k+1|t) &= A\nodState(k|t) + B\nodInput(k|t)), \quad \nodState(0|t) = \nodState(t),\quad &k\in \mathbb{I}_{[0,N-1]}, \\
     \nodVar(k|t) & =C\nodState(k|t) + D\nodInput(k|t), \hspace*{2.75cm}  &k\in \mathbb{I}_{[0,N-1]}, \\
    \begin{bmatrix}\nodState(k|t) &\nodInput(k|t)&\nodVar(k|t)\end{bmatrix}^\top&\in \X\times\U\times\Y(\nodDist(k|t)),  \hspace*{2.75cm}	  &k\in \mathbb{I}_{[0,N-1]}. \label{eq:OPF_OCP_constraints}
  \end{align}
  \end{subequations}
  
  Subsequently, we will refer to the  \OCP \eqref{eq:OPF_OCP} as the multi-stage \AC-\OPF problem. Moreover, we will analyze under which conditions \eqref{eq:OPF_OCP} satisfies the assumptions of Theorem \ref{thm:stab_EMPC}.

\section{Analysis of NMPC for multi-stage AC-OPF} \label{sec:results}
In analyzing the multi-stage \AC-\OPF \OCP \eqref{eq:OPF_OCP} it is crucial to observe that the (in general time-varying) disturbance $\nodDist(\cdot|t)$ enters the problem in the constraint \eqref{eq:OPF_OCP_constraints}. However, note that the conditions of Theorem \ref{thm:stab_EMPC} as such do not allow for time-varying problem data and/or constraints. Here, we present an analysis that exploits the fact that the disturbance $\nodDist(\cdot|t)$ \textit{only} enters \OCP \eqref{eq:OPF_OCP} through the constraint \eqref{eq:OPF_OCP_constraints}. As we will show it is this particular structure that enables conclusions about the receding-horizon solution to  \OCP \eqref{eq:OPF_OCP}. To this end, we consider in a first step the case that the disturbance $\nodDist$ changes only occasionally---i.e.  $\nodDist(t)\equiv const.$ for almost all $t\in\N$. In a second step we investigate the case that the disturbance $\nodDist(t)$ changes in each time-step sufficiently slowly.\footnote{We remark that in principle one could as well apply recent results on time-varying turnpike properties of \OCPs here \cite{Gruene17b, Gruene18b}. However, without prior knowledge of the disturbance sequence $\nodDist(\cdot)$ these conditions are hard to check. Thus we leave this point for future work.}

We begin our analysis with recalling obvious properties of the dynamics \eqref{eq:sys_OPF}.
\begin{lemma}[Invariance and reachability] \label{lem:invariance}
Consider system  \eqref{eq:sys_OPF} for some compact state constraint set $\X \subseteq \Rnx$.
Then the following statements hold:
\begin{enumerate}[label=(\roman*)]
\item For any non-empty set $\U$ with $0\in \inte\U$, any subset of the state constraint set $\X$ of system   \eqref{eq:sys_OPF} is controlled forward invariant.
\item If $\bar \X\subseteq \X$ is a path-connected subset of $\X$, then for any non-empty set $\U$ with $0\in \inte\U$ and any $x_0, x_1 \in \bar\X$, there exists a  feasible control sequence $u(k; x_0, x_1) = \{u(0), \dots, u(M-1)\}, u(k) \in \inte\U$ of finite length $M$ such that 
\begin{align*}
x(k; x_0, u(\cdot; x_0, x_1)) &\in \bar\X \quad \forall k \in \mathbb{I}_{[0, M]} \\
x(M; x_0, u(\cdot; x_0, x_1) &= x_1 .
\end{align*}
\item For $\U =\Rnu$ and any compact set $\X$,   any $x_1\in\X$ is 1-step reachable from any $x_0\in\X$.\eBox
\end{enumerate}
\end{lemma}
The proof follows without difficulties from the fact that in the multi-stage \OPF dynamics \eqref{eq:sys_OPF} we have  $A = I^{n_\nodState}$ and  $\rank(B) = {n_\nodState}$. For the sake of completeness it is given in the appendix.

In power systems, specifically in the context of \OPF problems, the considered objective functions are typically quadratic. Hence, we consider the following assumption.
\begin{assumption}[Quadratic stage costs] \label{ass:lq_cost}
The stage cost $\ell:\X\times\U\to\R$ is given by
\begin{align}\label{eq:lq_cost}
\ell(x,u) = x^\top Q x + u^\top R u + q^\top x + r^\top u. \vspace*{-.65cm}
\end{align}
\eBox
\end{assumption}
Recall that in  \eqref{eq:sys_OPF} the output equation reads $\nodVar = C\nodState + D\nodInput =: h(\nodState, \nodInput)$. Consider the set
\begin{equation} \label{eq:Yinvset}
 h^{-1}(\Y(\nodDist)) :=  \left\{[\nodState~\nodInput]^\top\,|\, \nodVar = C\nodState + D\nodInput \in \Y(\nodDist) \right\},
\end{equation}
which is the pre-image of $\Y(\nodDist)$ from \eqref{eq:Yset} with respect to the output equation \eqref{eq:sys_OPF_output}.
\begin{assumption}[Unconstrained optimal steady-state solutions] \label{ass:opf}
For all disturbances $\nodDist \in \D$ the steady-state minimizer
\begin{subequations}\label{eq:xsdusd}
\begin{align} 
(x_s(d),  u_s(d)) = \arg\min_{x, u}&\,\ell(x, u) \\
 \text{subject to }& \nonumber \\
x &= Ax+Bu\\
  \begin{bmatrix} x & u\end{bmatrix}^\top &\in (\X\times\U)\cap h^{-1}(\Y(\nodDist)).
\end{align}
\end{subequations}
exists and satisfies $(x_s(d), u_s(d)) \in\interior\left( (\X\times\U)\cap h^{-1}(\Y(\nodDist))\right)$. \eBox
\end{assumption}
As such the condition that $(x_s(d), u_s(d)) \in\interior\left( (\X\times\U)\cap h^{-1}(\Y(\nodDist))\right)$ is quite restrictive for real-world \OPF problems. Hence we will later comment on relaxing it.

\subsection{Constant Disturbances}
\begin{proposition}[Dissipativity of multistage \AC-\OPF \eqref{eq:OPF_OCP}] \label{prop:dissOPF}
Suppose Assumptions \ref{ass:lq_cost} and \ref{ass:opf} hold and let $\nodDist(\cdot)\in\D$ with $\nodDist(t) \equiv const$.

Then  system \eqref{eq:sys_OPF} is strictly dissipative with respect to $(\nodState_s(d), \nodInput_s(d)) \in  \interior $ $\Large(\left(\X\times\U\right)  \cap h^{-1}(\Y(d))\Large)$ and storage  $\lambda(x) = x^\top P x + p^\top x$ if and only if
  $Q >0$. \eBox
\end{proposition}
\begin{proof}
Observe that $\left(\X\times\U\right) \cap h^{-1}(\Y(d))$ is a controlled forward invariant subset of $\X\times\U$, cf. Lemma \ref{lem:invariance}. Hence, we will investigate dissipativity of \eqref{eq:sys_OPF} on $\X\times\U$ and then invoke Lemma \ref{lem:DI_subset}.

Similar to \cite[Lem 4.1]{Gruene18a} we consider $\lambda(x) = x^\top P x + p^\top x$ and the
 rotated stage cost 
\[
\tilde\ell(\nodState, \nodInput) = \ell(\nodState, \nodInput) - \ell(\nodState_s, \nodInput_s) + \lambda(x) - \lambda(Ax + Bu),
\]
which can be rewritten as
\[
\tilde\ell(\nodState, \nodInput) = x^\top\left(Q +P -A^\top P A\right)x + \varrho(x,u),
\]
where the term $ \varrho(x,u)$ collects all terms in $u$ and all non-quadratic terms in $x$. For system \eqref{eq:sys_OPF} to be strictly dissipative on $\left(\X\times\U\right)$  it needs to hold that
$
\tilde\ell(\nodState, \nodInput) \geq \alpha(\|x-x_s\|)
$
with $(\nodState_s, \nodInput_s) \in  \interior\left(\X\times\U\right)$. 
This implies that $\tilde\ell(\nodState, \nodInput_s)$ needs to have a strict unconstrained global minimum on 
$\X\times\U$. Since $\tilde\ell(\nodState, \nodInput_s)$ is quadratic in $x$, this implies that $Q +P -A^\top P A > 0$. By construction $A = I^{n_\nodState}$, hence we must have $Q >0$. Now, applying Assumption \ref{ass:opf} and Lemma \ref{lem:DI_subset} yields the assertion.
\end{proof}
The main insight of the last result can be reformulated as follows: only if $Q>0$ in $\ell$, will system \eqref{eq:sys_OPF} be dissipative with quadratic storage. 
Moreover, if $Q>0$, one can choose a linear storage  $\lambda(x) = p^\top x$ since $A = I^{n_\nodState}$ removes the matrix $P$ from the Lyapunov inequality  $Q +P -A^\top P A > 0$. However, in \OPF problems the generator costs (which are included in $\ell(x,u)$ are typically chosen to be strictly convex and quadratic in the respective part of $x$. Moreover, since usually one wants to avoid high values for the state of charge, it is reasonable to consider a convex penalization of the components of $x$ referring to the state of charge at the storage buses. Overall, $Q >0$ is not a severe restriction for the application at hand. 

Finally, we remark that Proposition \ref{prop:dissOPF} does not make any statement with respect to dissipativity of system \eqref{eq:sys_OPF}  for steady-state pairs located on the boundary of the constraints $\left(\X\times\U\right) \cap h^{-1}(\Y(d))$.\footnote{The difficulty in analyzing this case stems from the fact that it is tricky to characterize the boundary of the set $ h^{-1}(\Y(d))$ due to the nonlinearity of the power-flow manifold.} Such an undesirable case, however, would correspond to rather large disturbances pushing a power system to its limits.  
We leave a detailed analysis of this case for future work. 

Let $\Pi_x: (x,u) \mapsto x$ denote the projection from $\R^{n_x\times n_u}$ onto $\Rnx$. 
Consider the set 
\[
\X(d) := \Pi_x\left( (\X\times\U)\cap h^{-1}(\Y(\nodDist))\right) \subset \X,
\]
which is the projection  of the combined input-state constraints of \OCP \eqref{eq:OPF_OCP} onto $\R^{n_x}$. 

\begin{proposition}[Properties of economic \NMPC for \AC-\OPF]\label{prop:OPF_EMPC}
Consider the \NMPC scheme based on \OCP \eqref{eq:OPF_OCP}. Let Assumption \ref{ass:lq_cost} hold with $Q>0$, and let $\nodDist(t)\in\D$ with $\nodDist \equiv const$. For any $x_0$ and $d\in \D$, let $x_0$ and $x_s(\nodDist)$ be contained in a path-connected subset of $\X(d)$.

\begin{enumerate}[label=(\roman*)]
\item If Assumption \ref{ass:opf} holds with   $(x_s(d), u_s(d)) \in\interior\left( (\X\times\U)\cap h^{-1}(\Y(\nodDist))\right)$, then the \NMPC scheme based on \OCP \eqref{eq:OPF_OCP} satisfies the conditions of Theorem \ref{thm:stab_EMPC}. 
\item If   $(x_s(d), u_s(d)) \in\left( (\X\times\U)\cap h^{-1}(\Y(\nodDist))\right)$ (i.e. Assumption \ref{ass:opf} is relaxed accordingly) and for the chosen $Q$, \OCP \eqref{eq:OPF_OCP} is strictly dissipative with respect to $(x_s(d), u_s(d))$, 
then  the \NMPC scheme based on \OCP \eqref{eq:OPF_OCP} has the properties asserted in Theorem \ref{thm:stab_EMPC}. \eBox
\end{enumerate} 
\end{proposition}
\begin{proof}
Part (i): 
Recall that Assumptions \ref{ass:opf} yields that for all $d\in\D$ the optimal steady-state  pair $(x_s(d), u_s(d)) \in\interior\left( (\X\times\U)\cap h^{-1}(\Y(\nodDist))\right)$.
In order to apply  Theorem \ref{thm:stab_EMPC}, we need to show (i)--(iii) of Assumption \ref{ass:EMPC}: (i) Strict dissipativity follows from $Q>0$ and Proposition \ref{prop:dissOPF}. (ii) Exponential reachability follows from the requirement that $x_0$ and $x_s(\nodDist)$ are contained in a path-connected subset of $\X(d)$, cf. the proof of Part (ii) of Lemma \ref{lem:invariance}. (iii) It is trivial to see that \eqref{eq:sys_OPF} is 1-step reachable, cf.  Part (iii) of Lemma \ref{lem:invariance}.

Part (ii): Observe that Lemma \ref{lem:invariance} Part (ii) implies finite-time reachability of \eqref{eq:sys_OPF}. As Part (ii) suppose dissipativity of \OCP \eqref{eq:OPF_OCP} Corollary \ref{cor:stab_EMPC} can be applied.
\end{proof}
The last result shows that under quite mild technical assumptions the  \NMPC scheme based on \OCP \eqref{eq:OPF_OCP} will practically track the optimal steady state $x_s(d)$. Moreover, it is clear that the condition of $x_0$ and $x_s(\nodDist)$ being elements of a path-connected subset of $\X(d)$ is a kind of implicit reachability assumption, cf. the proof of Part (ii) of Lemma \ref{lem:invariance}. However, it deserves to be noted that path connectedness is merely sufficient and not necessary for reachability.

\begin{rema}[Occassionally varying disturbances]
Finally, we remark without in-depth elaboration that Proposition \ref{prop:OPF_EMPC} can be extended to the case of occasionally varying $d(t)$.  This extension requires ensuring the reachability of the $x_s(d(t))$ from $x(t)$ and that the next disturbance does not occur before the closed-loop system is converged to a sufficiently small neighborhood of $x_s(d(t))$. If such or similar conditions hold, one can expect that the economic \NMPC scheme  based on \OCP \eqref{eq:OPF_OCP}  will automatically track the optimal steady state.\eBox
\end{rema}

\subsection{Varying Disturbances}
However, in the context of power systems it is not to be expected that disturbances vary occasionally; rather they will vary continuously ($d(t) \not= d(t+1)$). Hence we turn to this case next, i.e. we discuss the case of time-varying disturbance sequences 
\begin{equation*}
d(\cdot) = d(0), d(1), \dots, d(N), \dots.
\end{equation*}
 Moreover, consider the set of states reachable from $x$ given the future disturbance $d$. In the extended $(x,u)$ space this set is given by 
\begin{equation}
\mbb{A}(x, d) = \left(f(x, \U)\times\U \right)\cap (\X\times\U)\cap h^{-1}(\Y(\nodDist)) \subseteq \Rnx\times\Rnu.
\end{equation}
This set is formed by the intersection of the reachable set neglecting state constraints---$f(x, \U)\times\U$---with the Cartesian product of state and input constraints---$\X\times\U$---and the mixed input-state constraints imposed by the power flow equations and the disturbance $d$---$h^{-1}(\Y(\nodDist))$. Note that we define the reachable set in the extended $(x,u)$ space as the definitions of $\Y(d)$ in \eqref{eq:Yset} and of $h^{-1}(\Y(\nodDist))$ in \eqref{eq:Yinvset} imply that the power flow equations constitute a coupled input-state constraint parametrized by $d$.

\begin{lemma}[Recursive feasibility of  \OCP \eqref{eq:OPF_OCP} for disturbance sequences $d(\cdot)$]\label{prop:reachOCP_seq}
If for all $x\in\X$ and all $d\in\D$
\begin{equation}\label{eq:reach_cond}
\mbb{A}(x, d) =  \left(f(x, \U)\times\U \right)\cap (\X\times\U)\cap h^{-1}(\Y(\nodDist)) \not=\emptyset
\end{equation}
holds.
Then, for any horizon $N$ and any  disturbance sequence $d(\cdot)\in\D$, if  \OCP \eqref{eq:OPF_OCP} is feasible for $t = 0$ and $x(0) \in \X$, it is feasible for all $t \in \N$. \eBox
\end{lemma}
\begin{proof}
Recall that $\mbb{A}(x, d)$ is the set of $(x,u)$ reachable from state $x$ for a given value of $d$. If for any combination of $x$ and $d$ this set is non-empty, this allows concluding that $x(t) \in  \Pi_x(\mbb{A}(x(t-1), d(t)))$ implies the existence of $u(t)\in\U$ such that $\begin{bmatrix}f(x(t), u(t)) & u(t) \end{bmatrix}^\top \in \mbb{A}(x(t), d(t+1))$.
\end{proof}
The above result establishes a connection between the variability of the disturbance $d\in\D$---i.e. the "size" of $\D$---and the reachability properties of the underlying system \eqref{eq:sys_OPF}. 
 However, we remark that---even neglecting input constraints---it is in general not possible to show reachability for arbitrarily large sets of values $d\in\D$. This follows from the structure of the power flow equations \eqref{eq:powerflow} which implies a physical upper limit for the power that can be transmitted over a line connecting two buses. 
 Excluding such pathological 
 cases, and given the reachability properties documented in Lemma \ref{lem:invariance}, it is evident that one has three main options to enforce that \eqref{eq:reach_cond} holds:
 either one ensures that the set $\D$ is sufficiently "small" (i.e. the disturbances cannot vary too much); or one supposes that the sequence $d(\cdot)$ varies slowly; 
or one ensures that the input constraint $\U$ is sufficiently "large" (i.e non-restrictive), respectively, the state constraint $\X$ is relaxed by adding more storage. 
 Considering the last option it deserves to be noted that besides the transmission capacity limit of the grid, the state constraints (which include active and reactive power injections by generators) put a limit on the "size" of the disturbance $d$ in terms of power demand. In other words, neglecting generator ramp constraints, generation capacity (plus storage) needs to satisfy load demands. 

Lemma \ref{prop:reachOCP_seq} highlights a sufficient condition for recursive feasibility in the \NMPC scheme based on \OCP \eqref{eq:OPF_OCP}. It does, however, not assert any kind of stability properties. 
Indeed, as the disturbance $d(\cdot)$ is allowed to vary in each time-step, stability of some kind of setpoint cannot be expected. 
However, using time-varying dissipativity notions, it might be possible to make further statements about closed-loop performance. This is subject of ongoing and future work.

\begin{proposition}[Performance bound for open-loop predictions] \label{prop:costBound}
Consider \OCP \eqref{eq:OPF_OCP} with $d(\cdot)\in\D$. Suppose that Assumption \ref{ass:opf} holds and that the conditions of Lemma \ref{prop:reachOCP_seq} are satisfied. Furthermore, let $\U = \Rnu$. 
Then, for all $N\in\N$ and all $x\in \X$, 
\[
V_N(x, t) \leq \sum_{k=0}^{N-1}\ell(x_s(d(k|t)),\, u_s(d(k|t))),
\]
whereby the sequence $(x_s(d(k|t)),\, u_s(d(k|t)))$ is generated by evaluating \eqref{eq:xsdusd} for all $d(k|t)$, $ k \in \{0, \dots, N-1\}$. \eBox
\end{proposition}
\begin{proof}
Observe that letting $\U = \Rnu$ any $x_1\in \X$ in the state constraint is reachable from any $x_0 \in \X$ in one step. Moreover, Assumption \ref{ass:opf} gives that the sequence $(x_s(d(t)),\, u_s(d(t)))$, $t\in\N$ is infinite-horizon feasible. Hence the optimal value function $V_N(x, t)$ is bounded by the performance generated by this sequence.
\end{proof}

A trivial consequence of the last result is that, whenever no storages are present and when the input constraints are sufficiently "large", the solution to  \OCP \eqref{eq:OPF_OCP}  corresponds to the sequence of optimal steady-state pairs
$(x_s(d(\cdot)),\, u_s(d(\cdot)))$ generated by evaluating \eqref{eq:xsdusd} for all $d(k|t), k \in \N$. 
Hence, one may expect/conjecture that the presence of storage will almost always improve closed-loop performance (as the presence of storages enlarges the feasible set). While a formal proof of such a property is beyond the scope the present paper, the simulation results presented in Section \ref{sec:example} indicate that this is indeed the case.

\section{Simulation Examples} \label{sec:example}
In this section we illustrate the findings from above by means of simulation examples. To this end, we consider a 5-bus system and the \IEEE-118 bus system shown in Figure~\ref{fig:5busCase} and Figure \ref{fig:118busCase} respectively. For both systems, we investigate two different load scenarios:  (i) we consider a piece-wise-constant load in  in order to study the behavior of the system at steady state and (ii) we consider a random demand sequence to investigate the system behavior in the time-varying case.
For the stage cost $\ell$ we consider the standard economic cost for active power generation \eqref{eq:lq_cost}; in both cases we add a small quadratic cost for the state of charge for all storages to ensure positive definiteness of $Q$. 
The sampling period is set to $\Delta = t_b = 1 \text{h}$ in all scenarios, which also serves as base time of for the here used p.u. normalization system.\footnote{The per-unit (p.u.) system is a standardized procedure of normalizing quantities in power systems engineering, cf. \cite{Grainger1994}. In this system, all powers are normalized to a corresponding basepower $S_b$ and a base voltage $V_b$. The base voltage varies for different buses and can be obtained from the \MATPOWER database. Here, we additionally define a basetime $t_b$ which is commonly not considered in the classical p.u. system for single-stage \OPF problems as it usually does not consider storage units. The base time is needed for expressing the the state-of-charge of all storages in per-unit, i.e. $s_{\mathrm{p.u.}}= s/{(S_b\cdot t_b)}$.}  Furthermore, we consider a base power of $S_b=100\,$MVA, and a prediction  horizon of $N=48$.
We simulate the closed-loop system  \eqref{eq:sys_closed_loop} from $t_0 =0$ up to $t_f = 48$ where the effective demands are given for $t_f + N = 96\,$. 
We consider a nominal setting for all simulations (perfect forecast for $d$ and no model-plant mismatch) and for all generators a maximum ramping capability of $10 \%$ of the generator's upper limit per hour. 
The problem data for both systems are retrieved from the \MATPOWER database\cite{Zimmerman2011}. 
Furthermore, we use CasADi 3.3.0\cite{Andersson12a} with \IPOPT\cite{Waechter2006} as underlying solver and \MATLAB R2016a on an Intel Core i7-4790 machine with $32\,$GB memory for all our simulations.  

We remark that the maximum computation time  for one \OCP  is less than one second for the 5-bus system, respectively, $20$s for the 118-bus system, while the usual sampling time for multi-stage \OPF is $15$min to $1$h. Hence, computation time does not pose a significant issue for medium-sized grids.  

\subsection{5-Bus System}

\begin{figure}[t]
	\centering
	\includegraphics[width=0.5\textwidth]{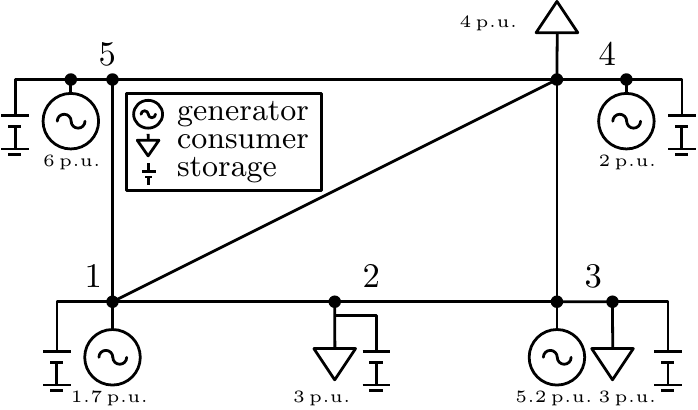}
	\caption{Modified 5-bus system with storage \cite{Li2010}.}
	\label{fig:5busCase}
\end{figure}

Figure~\ref{fig:5busCase} shows the considered  5-bus system, which is a slightly modified variant of \cite{Li2010}. We introduce energy storage with a capacity of $\bar e = 2\, \text{p.u}$ and maximum charging/discharging power of $\underline{p}^s = \bar p^s = 0.1 \, \text{p.u}$ at all buses. 
The parameters of the stage cost $\ell$ from \eqref{eq:lq_cost} are: ${Q=\operatorname{diag}\left (\operatorname{diag}\left (10,\; 11,\; 12,\; 13\right ),\,10^{-3}\cdot I^{9 \times 9}\right )\frac{\$}{{\text{h}\, (\text{p.u.})^2}}}$, 
$R=0^{18\times 18}$, ${q = 10^2\cdot \left (15,\; 30,\; 40,\; 10,\; 0^{1\times 9}\right )^\top\frac{\$}{{\text{h}\, (\text{p.u.})}}}$ and $r = {0^{18\times 1}}$.\footnote{The original case data from \cite{Li2010} considers linear cost for active power injections only. To ensure $Q\succ 0$ we add small quadratic cost terms here.}  All storages are initialized with $0.5\, \bar e$, i.e. $1\, \text{p.u.}$. 

We show our numerical results for two different disturbance scenarios: In the first one we assume an occasionally varying effective demands $d(t)$ (demands at each bus subtracted by the respective uncontrollable renewable feed-ins) with demand variations at bus 4. In the second case, we randomly generate realizations of an auto-correlated white Gaussian random process with time-varying mean to investigate the effects of time-varying disturbances.\footnote{The random demand sequences are derived based on the demand profile for Germany from the 23$^\text{\text{rd}}$ of July 2018 obtaining the reactive power demands by a constant power factor of $0.8$ for the 5-bus system. Note that as we consider the effective demand here, the demands can become negative in case the renewable regeneration exceeds the demand.} 

\begin{figure}[t] 
		\begin{subfigure}[c]{\textwidth}
		\includegraphics[trim=0 365 0 45, clip, width=0.95\textwidth]{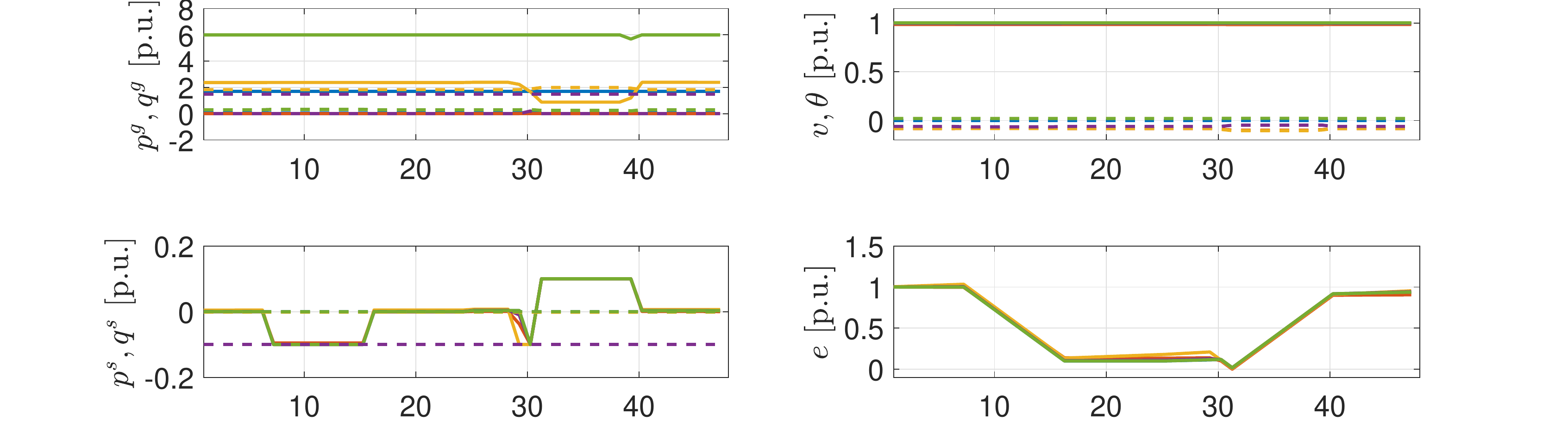}
		\captionsetup{justification   = raggedright,singlelinecheck = false,format=hang,width=.9\textwidth}
		\subcaption{Inputs $u(t)$, outputs $y(t)$, states $z(t)$ and $x(t)$ for all buses with $p^g(t),p^s(t),v(t), e(t)$ in solid and  $q^g(t),q^s(t), \theta(t)$ in dashed.}
		\label{fig:5busStepLoad1}
	\end{subfigure}
	\begin{subfigure}[c]{\textwidth}
		\includegraphics[trim=0 145 0 280, clip, width=0.95\textwidth]{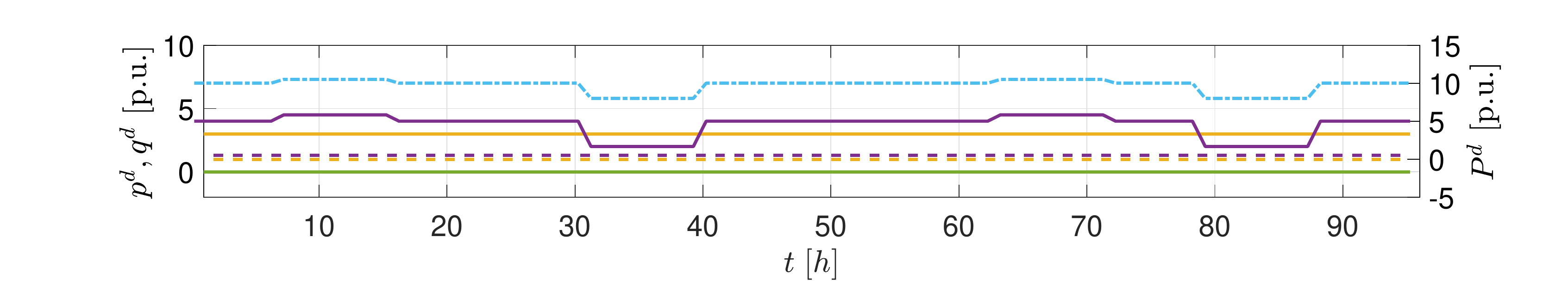}
		\captionsetup{justification   = raggedright,singlelinecheck = false,format=hang,width=.9\textwidth}
		\subcaption{Disturbances $d(t)$ with $p^d(t)$ in solid, $q^d(t)$  in dashed and aggregated active power demand $P^d(t)$ in dash-dotted.}
		\label{fig:5busStepLoad2}
	\end{subfigure}
	\begin{subfigure}[c]{\textwidth}
	\includegraphics[trim=0 15 0 420, clip, width=0.95\textwidth]{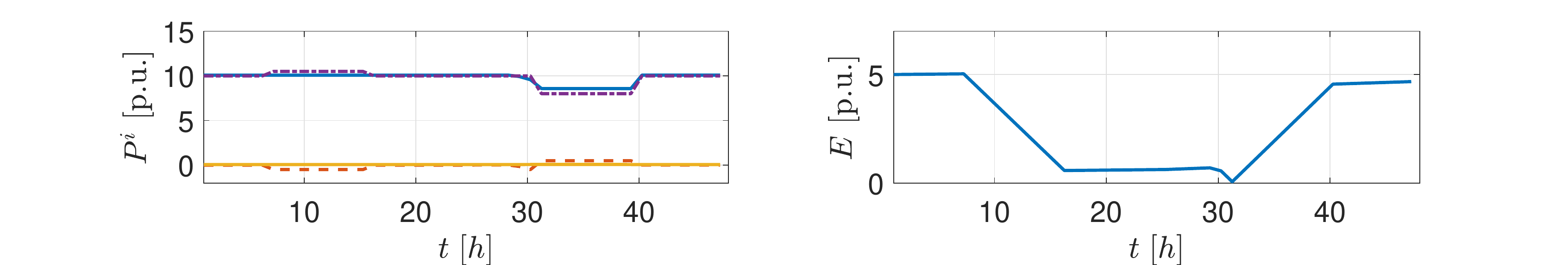}
		\captionsetup{justification   = raggedright,singlelinecheck = false,format=hang,width=.9\textwidth}
	\subcaption{Aggregated active power demands/infeeds/losses $P^i(t)$ with $i \in \{d,g,s,\mathrm{loss}\}$ and total state of charge $E(t)$ with $P^d(t)$ in dash-dotted, $P^g(t)$ in solid, $P^s(t)$ in dashed and $P^{\mathrm{loss}}(t)$ in yellow solid.}
	\label{fig:5busStepLoad3}
\end{subfigure}
	\caption{Closed-loop trajectories for the 5-bus system with occasionally varying effective active power demand.}
\label{fig:5busStepLoad}
\end{figure}
\begin{figure}[t] 
	\begin{subfigure}[c]{\textwidth}
		\includegraphics[trim=0 365 0 45, clip, width=0.95\textwidth]{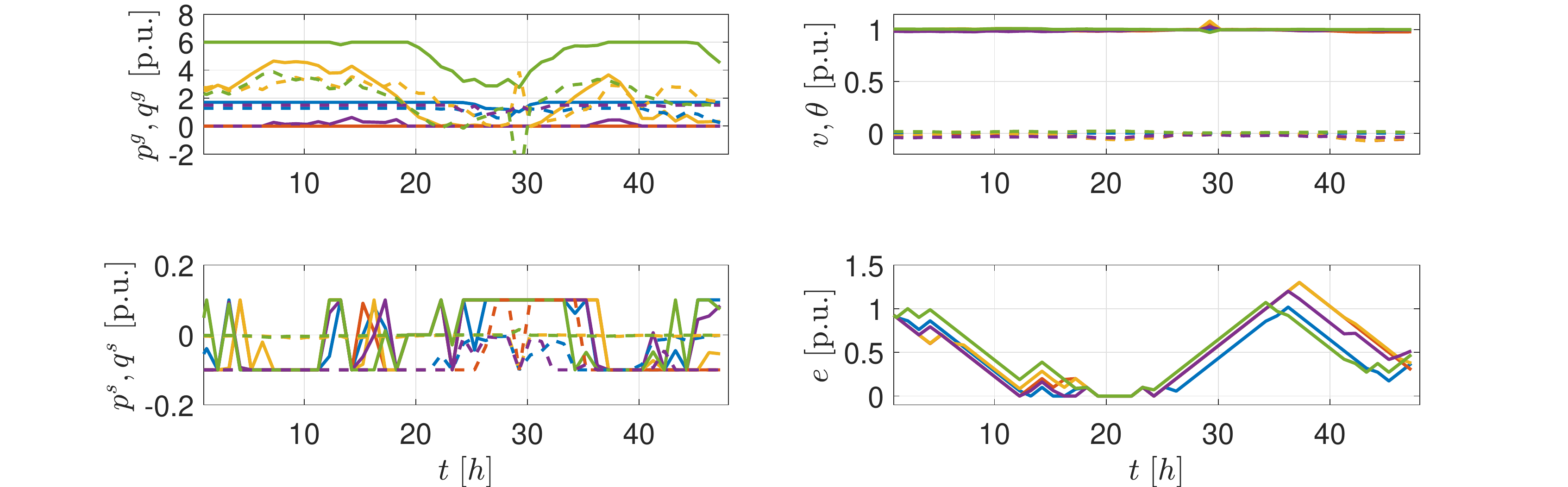}
		\captionsetup{justification   = raggedright,singlelinecheck = false,format=hang,width=.9\textwidth}
		\subcaption{Inputs $u(t)$, outputs $y(t)$, states $z(t)$ and $x(t)$ for all buses with $p^g(t),p^s(t),v(t), e(t)$ in solid and  $q^g(t),q^s(t), \theta(t)$ in dashed.}
		\label{fig:5busrandload1}
	\end{subfigure}
	\begin{subfigure}[c]{\textwidth}
		\includegraphics[trim=0 145 0 280, clip, width=0.95\textwidth]{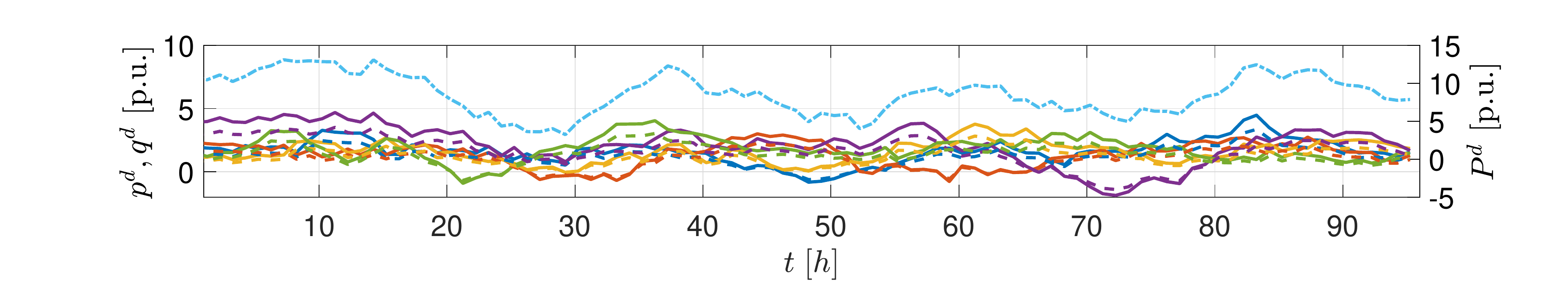}
		\captionsetup{justification   = raggedright,singlelinecheck = false,format=hang,width=.9\textwidth}
		\subcaption{Disturbances $d(t)$ with $p^d(t)$ in solid, $q^d(t)$  in dashed and aggregated active power demand $P^d(t)$ in dash-dotted.}
		\label{fig:5busrandload2}
	\end{subfigure}
	\begin{subfigure}[c]{\textwidth}
		\includegraphics[trim=0 15 0 420, clip, width=0.95\textwidth]{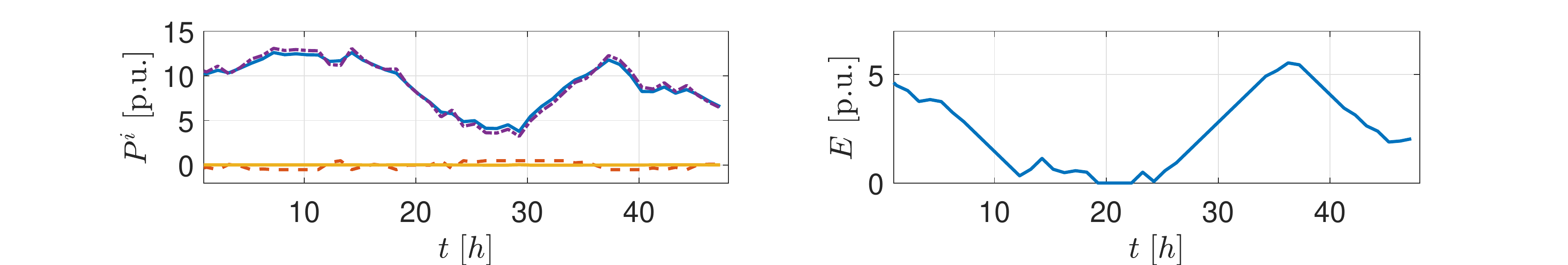}
		\captionsetup{justification   = raggedright,singlelinecheck = false,format=hang,width=.9\textwidth}
		\subcaption{Aggregated active power demands/infeeds/losses $P^i(t)$ with $i \in \{d,g,s,\mathrm{loss}\}$ and total state of charge $E(t)$ with 	$P^d(t)$ in dash-dotted, $P^g(t)$ in solid, $P^s(t)$ in dashed and $P^{\mathrm{loss}}(t)$ in yellow solid.}
		\label{fig:5busrandload3}
	\end{subfigure}
		\caption{Closed-loop trajectories for the 5-bus system with random effective active power demand.}
	\label{fig:5busrandload}
\end{figure}

For the sake of intuitive interpretation of our numerical results, we define aggregated values for active power generation/demands, the storage powers and the state of charge as
\begin{align*}
P^g(t) = \sum_{l \in \mathcal{G}} p^g_l(t), \qquad
P^d(t) = \sum_{l \in \mathcal{N}} p^d_l(t), \qquad
P^s(t) = \sum_{l \in \mathcal{S}} p^s_l(t), \qquad
E(t) = \sum_{l \in \mathcal{S}} e_l(t), 
\end{align*}
where $P^d(t)$ is the total active power demand, $P^g(t)$ is the total active power generation, $P^s(t)$ is the total active power storage feed-in and $E(t)$ is the aggregated state-of-charge of all storage devices at time $t$. 
Furthermore, $P^{\mathrm{loss}}(t) = P^g(t) - P^d(t)-P^s(t)$ are the total grid losses. 

Figure \ref{fig:5busStepLoad} and Figure \ref{fig:5busrandload} show the resulting input sequences $u(t)$, output sequences $u(t)$, demands $d(t)$, and closed-loop trajectories for the differential states $x(t)$ and algebraic states $z(t)$ over $48\text{h}$. Herein, Figure \ref{fig:5busStepLoad1} shows active and reactive power generations for all generators and storages as well as voltage magnitudes, voltage angles and state of charge of the batteries for all buses of the 5-bus system. Figure \ref{fig:5busStepLoad2} shows effective active and reactive power power demands for all buses and furthermore the aggregated active power demand $P^d(t)$. Finally, Figure \ref{fig:5busStepLoad3} depicts aggregated active power demands, infeeds and losses as well as the aggregated state of charge for all storages.

In case of occasionally varying loads (Figure~\ref{fig:5busStepLoad}) one  can observe the following: In high-demand situations ($P^d(t)$ and $P^s(t)$ for $t\in \mathbb{I}_{[8,\; 15]}$) energy from storages is used to cover the demand while in low-demand situations ($P^d(t)$ and $P^s(t)$ for $t\in \mathbb{I}_{[30,\; 40]}$) energy is stored leading to improved total operating cost. The reason for this cost-improvement is that we consider  small cost coefficients for storing energy while the cost for active power generation increases quadratically with $p^g$. Hence, it is economically beneficial to charge the storages in low-demand situations and use this energy in high-demand situations. Secondly, we can observe generator ramp constraints becoming active at the demand steps at $t=30$ and $t=40$.  
This furthermore underlines the importance of energy storage not only for load-shifting but also for providing fast ramping capabilities  (cf. $p^s(t)$ at $t=30$ ramping from $-0.2\,$p.u. to $0.2\,$p.u. in one time step), which are important in practice.\footnote{Especially in case of a high-share of renewables, conventional power plants struggle to provide fast generation ramps due to their limited transient performance induced by the underlying slow dynamics of firing and steam generation. 
} 
Figure~\ref{fig:5busrandload} shows the same quantities as Figure~\ref{fig:5busStepLoad} but with randomly generated demands. The behavior of energy shifting from low-demand to high-demand situations ($P^d(t)$ and $P^s(t)$ for $t\in \mathbb{I}_{[20,\; 35]}$) and active generator ramps  can also be observed. 
Finally, summing up the curves for $P^i$ in Figure~\ref{fig:5busStepLoad3} and Figure~\ref{fig:5busrandload3} illustrates that energy conservation 
\[
P^g(t) - P^d(t)-P^s(t) -P^{\text{loss}}(t)=0,
\]
which is implicitly considered via the \AC power flow equations \eqref{eq:powerflow}, is satisfied for all  $t\in \mathbb{I}_{[0,\; 48]}$.

\subsection{118-Bus System}

\begin{figure}[h]
	\centering
	\includegraphics[width=0.7\textwidth]{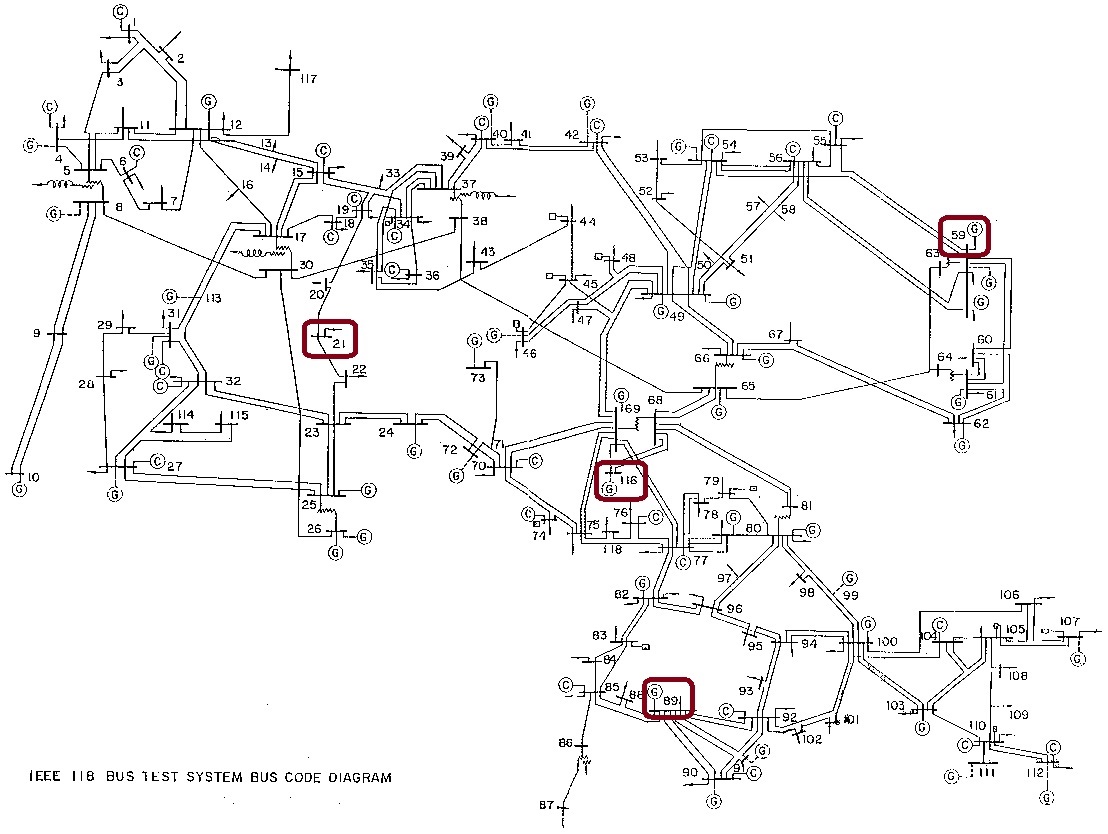}
	\caption{\IEEE 118-bus system with selected nodes (red) for numerical evaluation in Figure \ref{fig:118busStepLoad} and Figure \ref{fig:118busRandLoad} \cite{IEEE118bus}.  } 
	\label{fig:118busCase}
\end{figure}

\begin{figure}[t] 
	
	\begin{subfigure}[c]{\textwidth}
		\includegraphics[trim=0 365 0 45, clip, width=0.95\textwidth]{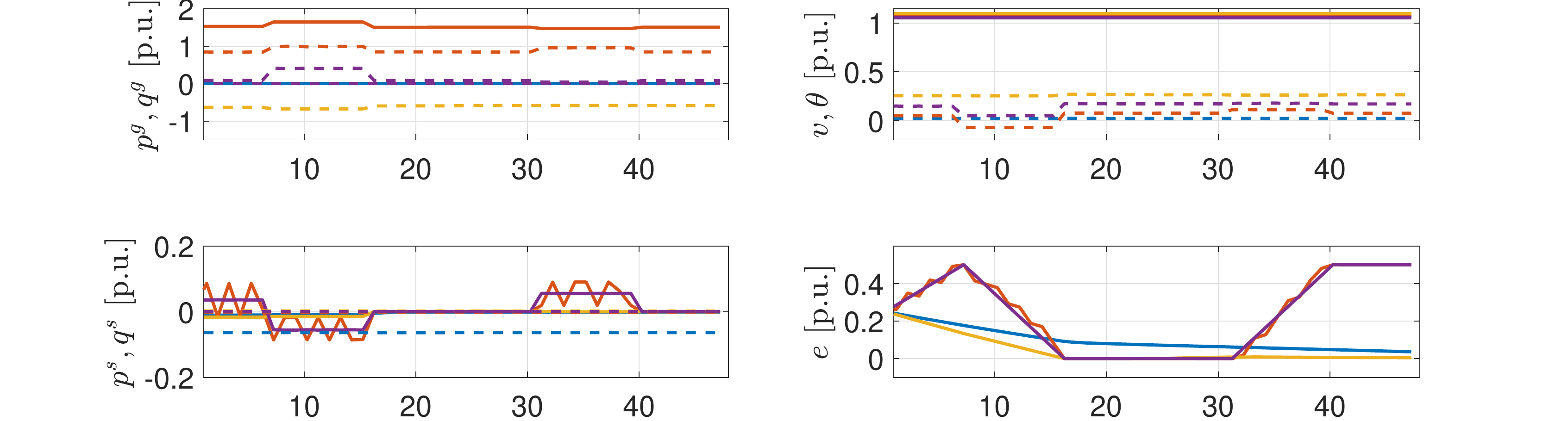}
		\captionsetup{justification   = raggedright,singlelinecheck = false,format=hang,width=.9\textwidth}
			\subcaption{Inputs $u(t)$, outputs $y(t)$, states $z(t)$ and $x(t)$ for buses $l \in \{21, 59, 89, 116\}$ with $p^g(t),p^s(t),v(t), e(t)$ in solid and  $q^g(t),q^s(t), \theta(t)$ in dashed.}
		\label{fig:118busStepLoad1}
	\end{subfigure}
	\begin{subfigure}[c]{\textwidth}
		\includegraphics[trim=0 180 0 350, clip, width=0.95\textwidth]{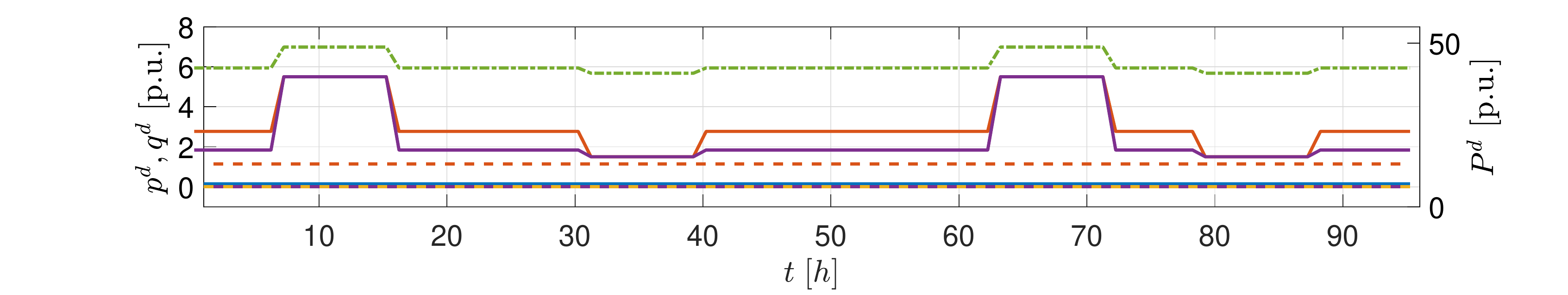}
		\captionsetup{justification   = raggedright,singlelinecheck = false,format=hang,width=.9\textwidth}
		\subcaption{Disturbances $d(t)$  for buses $l \in \{21, 59, 89, 116\}$ with $p^d(t)$ in solid, $q^d(t)$  in dashed and aggregated active power demand $P^d(t)$ in dash-dotted.}
		\label{fig:118busStepLoad2}
	\end{subfigure}
	\begin{subfigure}[c]{\textwidth}
		\includegraphics[trim=0 15 0 510, clip, width=0.95\textwidth]{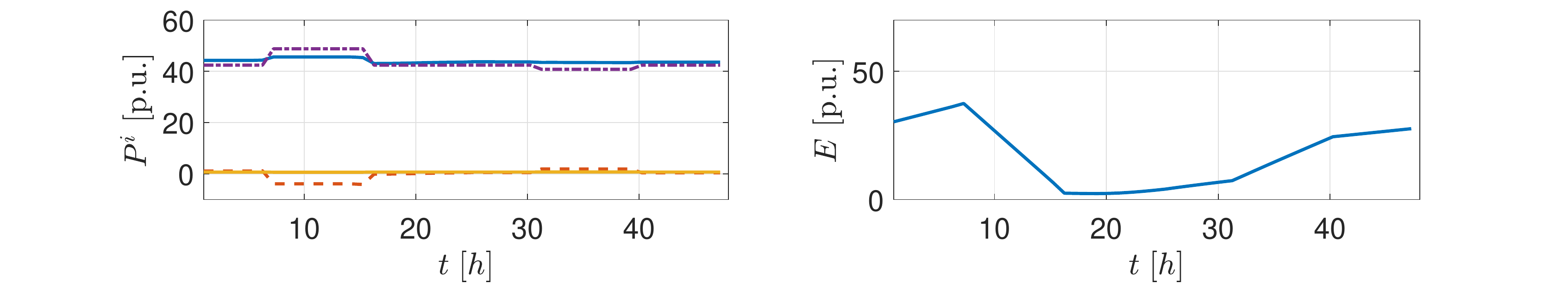}
		\captionsetup{justification   = raggedright,singlelinecheck = false,format=hang,width=.9\textwidth}
		\subcaption{Aggregated active power demands/infeeds/losses $P^i(t)$ with $i \in \{d,g,s,\mathrm{loss}\}$ and total state of charge $E(t)$ with $P^d(t)$ in dash-dotted, $P^g(t)$ in solid, $P^s(t)$ in dashed and $P^{\mathrm{loss}}(t)$ in yellow solid.}
		\label{fig:118busStepLoad3}
	\end{subfigure}
	\caption{Closed-loop trajectories for the \IEEE 118-bus system with occasionally varying effective active power demand.}
	\label{fig:118busStepLoad}
\end{figure}
\begin{figure}[t] 
		\begin{subfigure}[c]{\textwidth}
		\includegraphics[trim=0 365 0 45, clip, width=0.95\textwidth]{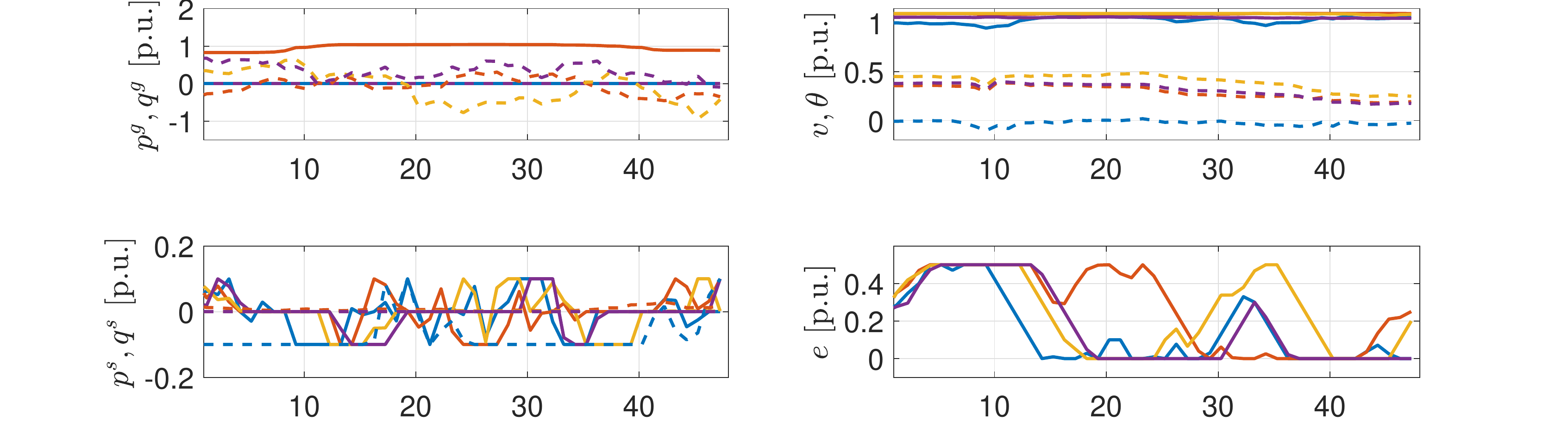}
		\captionsetup{justification   = raggedright,singlelinecheck = false,format=hang,width=.9\textwidth}
			\subcaption{Inputs $u(t)$, outputs $y(t)$, states $z(t)$ and $x(t)$ for buses $l \in \{21, 59, 89, 116\}$ with $p^g(t),p^s(t),v(t), e(t)$ in solid and  $q^g(t),q^s(t), \theta(t)$ in dashed.}
		\label{fig:118busRandLoad1}
	\end{subfigure}
	\begin{subfigure}[c]{\textwidth}
		\includegraphics[trim=0 180 0 350, clip, width=0.95\textwidth]{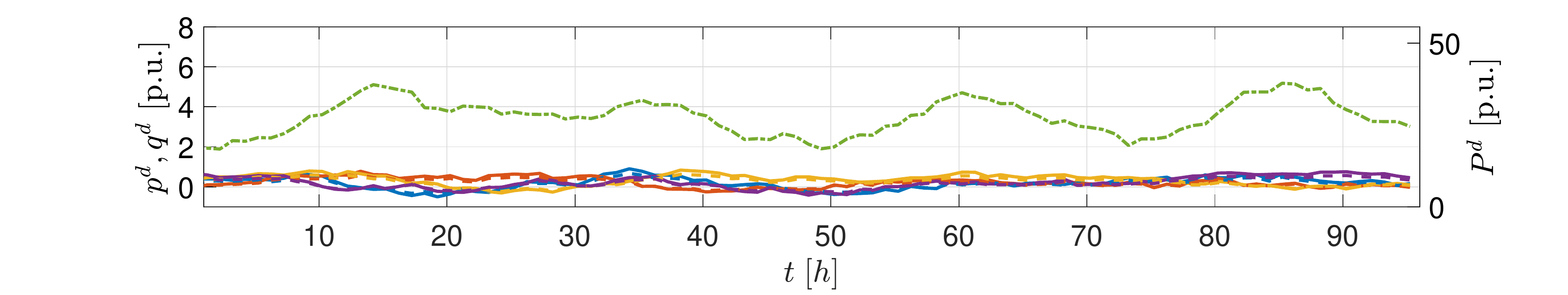}
		\captionsetup{justification   = raggedright,singlelinecheck = false,format=hang,width=.9\textwidth}
		\subcaption{Disturbances $d(t)$  for buses $l \in \{21, 59, 89, 116\}$ with $p^d(t)$ in solid, $q^d(t)$  in dashed and aggregated active power 	demand $P^d(t)$ in dash-dotted.}
		\label{fig:118busRandLoad2}
	\end{subfigure}
	\begin{subfigure}[c]{\textwidth}
		\includegraphics[trim=0 15 0 510, clip, width=0.95\textwidth]{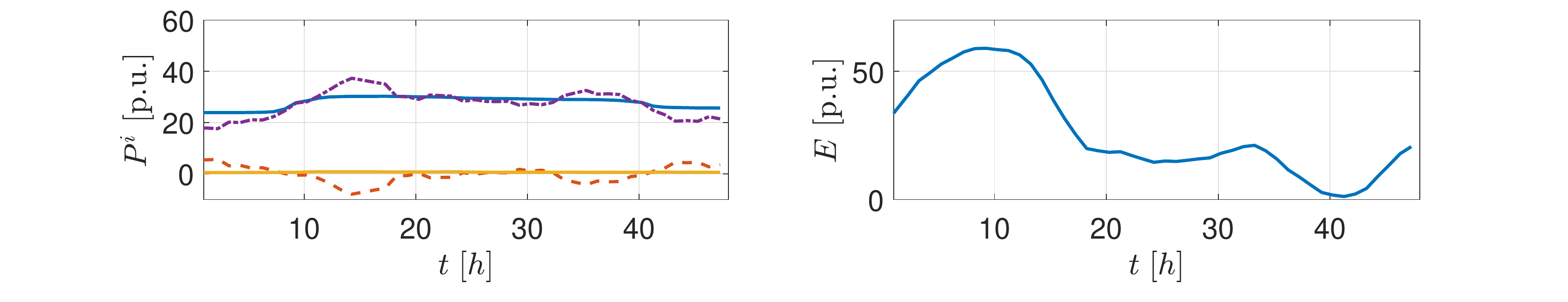}
		\captionsetup{justification   = raggedright,singlelinecheck = false,format=hang,width=.9\textwidth}
		\subcaption{Aggregated active power demands/infeeds/losses $P^i(t)$ with $i \in \{d,g,s,\mathrm{loss}\}$ and total state of charge $E(t)$ with $P^d(t)$ in dash-dotted, $P^g(t)$ in solid, $P^s(t)$ in dashed and $P^{\mathrm{loss}}(t)$ in yellow solid.}	
		\label{fig:118busRandLoad3}
	\end{subfigure}
	\caption{Closed-loop trajectories for the \IEEE 118-bus system with random effective active power demand.}
	\label{fig:118busRandLoad}
\end{figure}

Next, we consider the \IEEE 118-bus system depicted in Figure \ref{fig:118busCase}. 
Due to the large number of buses in this system, we show our results for selected buses $l \in \{21, 59, 89, 116\}$. Similarly to the 5-bus system from above, we introduce storage units with a capacity of $0.5\, \text{p.u}$ and maximum charging/discharging power of $0.1 \, \text{p.u}$ to all buses. 
Furthermore, we also consider occasionally varying demands at bus 59 and 116 with results shown in Figure~\ref{fig:118busStepLoad} and consider randomly generated demands as above in Figure~\ref{fig:118busRandLoad}. 
The parameters for the stage-cost \eqref{eq:lq_cost} are obtained from the \MATPOWER database with additional small quadratic coefficients of $10^{-3}\frac{\$}{{\text{h}\, (\text{p.u.})^2}}$ for $q^g$ and $e$ in $Q$ in order to ensure $Q\succ 0$. 
Similar to the previous case, we choose $R=0^{344\times 344}$ and $r=0^{344\times 1}$.
All effects from the previous 5-bus case (i.e. energy shifting from low-demand to high demand, active generator ramps) can also here be observed in both load cases (Figure~\ref{fig:118busStepLoad} and Figure~\ref{fig:118busRandLoad}). 
In addition to the 5-bus case from above, we observe active upper bounds for voltages. This highlights the importance of \AC \OPF over \DC \OPF as among others voltage limits can not be considered in the \DC case.

\subsection{Performance Comparison with Tracking MPC}

In this section we compare the close-loop performance of the above economic \MPC scheme with two different tracking \MPC schemes inspired from schemes in the literature to highlight the benefits of using \emph{economic} \NMPC.
The first approach to which we compare our results aims at tracking $50\%$  state of charge of the batteries (storages)
in order to keep a maximum of flexibility for future high-demand/high-load situations similar as in  \cite{Giorgio2017}. 
The state of charge tracking cost is
\[
\ell_b(x,u)= \sum_{l \in \mathcal{S}} c_l\,(e_l(t) - \tilde e_l)^2
\]
where we choose $\tilde e_l= 0.5\, \bar e_l$ and $c_l=1$.
Furthermore, an alternative  tracking \MPC formulation for multi-stage \OPF is
\[
\ell_{t}(x,u)=  \left (x(t)-x_{s}(t)\right )^\top Q\, \left (x(t)-x_{s}(t)\right ) + \left (u(t)-u_{s}(t)\right )^\top R\, \left (u(t)-u_{s}(t)\right )^\top.
\]
This formulation aims at tracking the solution sequence of the single-stage problems $x_s(d(t)), u_s(d(t))$ obtained from \eqref{eq:xsdusd}. Here, we choose $Q = I^{n_x}$ and $R=0.1\cdot$$ \diag(0^{2|\mathcal{G}|},I^{2|\mathcal{S}|})$.\footnote{Single-stage sequence here means solving individual \OPF problem for each time step $t$ and given $d(t)$ neglecting storages and generator-ramp constraints. Note that the sequence $x_{s}(\cdot)$, is as such varying with time since $p^g$ and $q^g$ change over $k$.}

Let $x_{e}(\cdot), u_{e}(\cdot)$ be the closed-loop trajectories obtained by using the economic cost $\ell$, let $x_{b}(\cdot), u_{b}(\cdot)$ be the closed-loop trajectories obtained by the state of charge tracking cost $\ell_b$, let $x_{t}(\cdot), u_{t}(\cdot)$ be the closed-loop trajectories for the single-stage tracking cost $\ell_{t}$ and let $d(\cdot)$ be the disturbance sequence. Then for the time span $\mbb{I}_{[t_0, t_f]} \subset\N$ closed-loop performances with respect to the economic cost $\ell$ are defined as
\[
J_{e} := \sum_{t= t_0}^{t_f} \ell(x_{e}(t), u_{e}(t)), \qquad
J_{b} := \sum_{t= t_0}^{t_f} \ell(x_{b}(t), u_{b}(t)),  \quad \text{and} \quad 
J_{t} := \sum_{t= t_0}^{t_f} \ell(x_{t}(t), u_{t}(t)).
\]
Note that we evaluate the economic generator cost $\ell$ along the closed-loop trajectories generated with $\ell, \ell_s$ and $\ell_t$ in order to obtain a performance comparison in terms of the underlying economic objective.
The performance of the sequence of optimal single-stage solutions is given by
\[
J_{s} := \sum_{t= t_0}^{t_f} \ell(x_s(d(t)), u_s(d(t))),
\]
where as before $x_s(d(t)), u_s(d(t))$ are obtained from \eqref{eq:xsdusd}.

Table~\ref{tab:costs} compares the closed-loop performances to the quasi steady-state for the occasionally varying disturbances and the randomly generated disturbances.
Observe that  for all cases  the performance of the sequence of optimal single-stage solutions is indeed worse than the closed-loop performance with storages and ramp constraints. This is due to the absence of storages. Not surprisingly, tracking the sequence  of optimal single-stage solutions  leads to decreased performance compared to economic \MPC. Furthermore, pure state-of-charge tracking results in substantial loss in the economic performance. 

\begin{table}[t]
	\centering
	\begin{tabular}{l||rrrr}
		case	&   $J_{e}\,[\text{k}\$]$  & $J_s\,[\text{k}\$]$ & $J_{t}\,[\text{k}\$]$ & $J_{b}\,[\text{k}\$]$  \\ 
		\hline 
		5-bus (rand.)	& 910.5   & 925.5 & 931.8 & 1,225.8 \\  
		5-bus (step)	& 973.1    & 976.7 & 978.6 & 1,265.1\\ 
		118-bus (rand.)	& 3,546.8  & 3,606.6 &3,606.8& 4,748.6\\ 
		118-bus (step)	& 6,377.5 & 6,387.4 & 6,620.0 & 7,350.6\\ 
	\end{tabular} 
	\caption{Performance comparison of \NMPC and quasi steady state with ramp constraints for a  $48\text{h}$ simulation.}
	\label{tab:costs}
\end{table}

\section{Conclusions and Outlook} \label{sec:conclusion}
This paper has presented an economic \NMPC approach to receding-horizon multi-stage \AC \OPF. 
It appears 
to be the first attempt towards an economic \NMPC-inspired analysis considering the full \AC power flow equations as  constraints of the optimization.

Specifically, we analyzed the discrete-time \OCPs arising in the considered multi-stage \AC-\OPF setting using a dissipativity notion stemming from economic \NMPC. Moreover, we have discussed the recursive feasibility properties of these \OCPs, which is of importance for receding-horizon (or \NMPC) solutions. The key step in our approach is the reformulation of the power-flow  equations in terms of set-based (projected) constraints this way avoiding the tedious technicalities of non-uniqueness of solutions to the power flow equations. For the case of constant disturbances/loads we have shown that dissipativity of the  \AC-\OPF \OCP is directly implied by  strict convexity of the usually quadratic stage cost with respect to the state variables. Moreover, we have analyzed the recursive feasibility properties of the proposed \NMPC schemes.   Finally, we have illustrated our findings drawing upon systems with $5$ and $118$ buses. The simulations underpin that using energy storage helps decreasing economic operating costs of power systems.

However, the present paper is merely a first step aiming at the transfer of the recent progress on economic \NMPC to power systems. Indeed the verification of time-varying dissipativity notions for the \AC-\OPF \OCP is still an open problem. Finally, in any real-world application the forecasts of the disturbances/loads will inevitably be surrounded by substantial stochastic uncertainties. Hence, further research on stochastic economic \NMPC for multistage \OPF is needed. 

\section{Appendix} \label{sec:app}
\begin{proof}[Proof of Lemma \ref{lem:invariance}] The proof exploits that in multi-stage \OPF problems  $A = I^{n_\nodState}$ and  $\rank(B) = {n_\nodState}$.
Part (i): Observe that each element of the state constraint set $\X$ is a steady state of \eqref{eq:sys_OPF} since $A = I^{n_\nodState}$.
Part (ii): W.l.o.g. we restrict $\U$ to some open neighborhood $\mcl{B}_\rho(0), \rho >0$ of $0$. By assumption the points $x_0$ and $x_1$ are connected by a continuous curve $\gamma:s\in [0,1] \to  \X$ such that $\gamma(0) = x_0$ and $\gamma(1) = x_1$.
Now,  consider two points $\gamma(s+\varepsilon)$ and $\gamma(s) $ with  $\varepsilon > 0$. An input $u(s, \varepsilon)$ transferring $x$ from  $\gamma(s+\varepsilon)$ to $\gamma(s) $ has to satisfy
\[
\gamma(s+\varepsilon) = A\gamma(s) + Bu(s, \varepsilon).
\] 
This equation can be solved for $u(s,\varepsilon)$ using the Moore-Penrose inverse of $B$ as follows
\[
u(s,\varepsilon) =B^\top\left(BB^\top\right)^{-1}(\gamma(s+\varepsilon) - \gamma(s))
\]
since $n_u \geq n_x$ and $\rank B = n_x$.
Due to continuity of $\gamma$, one may choose $\varepsilon >0$ such that $\frac{1}{\varepsilon} = M \in \N$ and, for all $s\in [0, 1-\varepsilon]$, it holds that $u(s,\varepsilon) \in \mcl{B}_\rho(0) \subset \U$. Moreover, since $\gamma \in \X$, the sequence
$u(s,k\varepsilon), k\in \mbb{I}_{[0, M-1]}$ transfers $x_0$ to $x_1$ in $M$ steps without leaving $\bar\X$.
 Part (iii) is an obvious implication of Part (ii).
\end{proof}
\clearpage

\printbibliography

\end{document}